\documentclass{llncs}
\usepackage[T1]{fontenc}

\usepackage{amsmath,amsfonts,amssymb,amstext}
\usepackage{graphicx}
\usepackage[usenames,dvipsnames]{xcolor}
\usepackage{subcaption}
\usepackage{url}
\usepackage{enumerate}
\usepackage{multirow}
\usepackage{rotating}
\usepackage{verbatim}
\usepackage{mathabx}
\usepackage{soul}
\usepackage{tipa}
\usepackage{arcs}
\usepackage{mathdots}
\usepackage{algorithm}
\usepackage{algorithmic}
\usepackage[colorlinks=true,linkcolor=MidnightBlue,citecolor=DarkGreen,urlcolor=Maroon]{hyperref}

\usepackage{thmtools} 
\definecolor{DarkGreen}{rgb}{0.0, 0.5, 0.0}

\DeclareFontFamily{OMX}{yhex}{}
\DeclareFontShape{OMX}{yhex}{m}{n}{<->yhcmex10}{}
\DeclareSymbolFont{yhlargesymbols}{OMX}{yhex}{m}{n}
\DeclareMathAccent{\wideparen}{\mathord}{yhlargesymbols}{"F3}

\newcommand{\naturals}{\mathbb{N}}
\newcommand{\integers}{\mathbb{Z}}
\newcommand{\reals}{\mathbb{R}}
\newcommand{\bone}{\textbf{1}}

\newcommand{\coss}[1]{\cos\left(#1\right)}

\newcommand{\sinn}[1]{\sin\left(#1\right)}

\newcommand{\norm}[1]{\left\lVert#1\right\rVert}

\newcommand{\ignore}[1]{}
\newcommand{\ignoreproof}[1]{}

\newcommand{\dd}{\,\mathrm{d}}
\newcommand{\msp}{\textsc{MSP}}
\DeclareMathOperator*{\argmin}{arg\,min}

\title{
Spirals and Beyond:\\
Competitive Plane Search with Multi-Speed Agents\thanks{This paper is the full version of a work accepted to the 17th Latin American Theoretical Informatics Symposium (LATIN 2026), held April~13--17,~2026, in Florianópolis, Brazil~\cite{GJMSpiralLatin26}.}
}

\author{
Konstantinos Georgiou\orcidID{0009-0007-9677-341X}\thanks{Research supported in part by TMU Faculty of Science Dean’s Research Fund 2024, and by NSERC of Canada. \texttt{konstantinos@torontomu.ca}}
\and
Caleb Jones\thanks{Research supported in part by NSERC of Canada. \texttt{caleb.w.jones@torontomu.ca}}
\and
Matthew Madej
}

\institute{
Department of Mathematics,
Toronto Metropolitan University,
Toronto, ON, Canada
}

\date{}  

\begin{document}
\maketitle
\thispagestyle{empty}

\begin{abstract}
We consider the problem of minimizing the worst-case search time for a hidden point target in the plane using multiple mobile agents of differing speeds, all starting from a common origin. The search time is \emph{normalized} by the target's distance to the origin, following the standard convention in competitive analysis. The objective is to minimize the maximum such normalized time over all possible target locations, known as the \emph{search cost}. As a foundation for our main results, we extend the known result for a single unit-speed agent, which achieves a provably optimal search cost of approximately $\mathcal{U}_1 = 17.28935$ via a logarithmic spiral, to the setting of $n$ unit-speed agents. We design a symmetric spiral-based algorithm in which each agent follows a logarithmic spiral offset by \emph{equal} angular phases. A key feature of this construction is that the resulting search cost is independent of which agent finds the target. We provide a closed-form upper bound $\mathcal{U}_n$ for the search cost in this setting, which forms the basis for our general result.

Our main contribution is an upper bound on the worst-case normalized search time for $n$ agents with \emph{arbitrary speeds}. 
We describe an algorithmic framework that selects a subset of the agents and assigns spiral-type trajectories with \emph{speed-dependent angular offsets}, ensuring that the search cost is independent of which agent reaches the target. A key corollary shows that a team of $n$ multi-speed agents, with the fastest agent having speed 1, 
can improve on the constructed $k$-unit-speed bound $\mathcal{U}_k$
 if the geometric mean of their speeds exceeds $\mathcal{U}_n / \mathcal{U}_k$. An implication of this result is that slow agents may be excluded from the spiral search if their presence lowers the geometric mean too much. This motivates the study of non-spiral algorithms that can do better in such cases. We establish new upper bounds for point search in \emph{cones} and \emph{conic complements} using a single unit-speed agent, which may be of independent interest. 
These bounds are then used to design hybrid strategies that combine spiral and directional search, and we show that these can outperform the previously derived spiral-based algorithms for a nonempty interval of low agent speeds.
 Thus, the common-expansion-factor Off-Set family is not always optimal in the general multi-speed setting. This does not exclude other spiral families.
\end{abstract}

\newpage
\thispagestyle{empty}
\tableofcontents

\newpage
\pagenumbering{arabic}

\section{Introduction}

The \emph{lost-in-a-forest problem}, introduced by Bellman in 1955~\cite{bellman1956minimization}, asks how to navigate from an unknown location in a known region to its boundary while minimising expected travel time. It remains open and foundational to search theory. Many variants have since been studied, yielding notable upper bounds and generally loose lower bounds. We review selected results in Section~\ref{sec: related work}.

A central variant involves a unit-speed agent searching for a hidden \emph{shoreline} in the plane, modelled as an \emph{infinite line}. The case where the shoreline distance is known was solved in 1957~\cite{isbell1957optimal}, one of the few known optimal results. Until very recently, the unknown-distance case remained open, with the logarithmic spiral conjectured to be optimal~\cite{baeza1988searching}. The best published unconditional lower bound was $6.3972$~\cite{baeza1995parallel}, while the stronger cyclic-path result of~\cite{langetepe2012searching} concerns axis-parallel shoreline search.
Two recent preprints claim stronger results. The first gives an unconditional lower bound of $12.5937096701\ldots$~\cite{temerev2026lower}, while the second claims a computer-assisted proof that the logarithmic spiral is optimal~\cite{temerev2026optimal}. 

We consider a related variant introduced in~\cite{gal2010search}, where the target is a \emph{point} in the plane. A point is \emph{exposed} by an agent if it lies on the line segment between the agent’s current position and the origin. The goal is to design a trajectory that minimises the worst-case target exposure time, normalized by the target's distance from the origin. The logarithmic spiral was first shown to be optimal among monotone periodic trajectories~\cite{gal2010search}, and later proven unconditionally optimal in~\cite{langetepe2010optimality}.

We extend the upper bound results to a team of \emph{multi-speed agents}, a rarely studied generalization. First, we construct natural logarithmic spiral trajectories with speed-dependent offsets, which are non-intersecting whenever every participating agent exposes new points and for which all such agents incur equal cost. We also identify speed regimes where agents do or do not contribute to exposing new points. As a corollary, we relate multi-speed performance to that of unit-speed agents, and give a condition, based on the geometric mean of the fastest agents’ speeds, for effective contribution. Second, motivated by the limitations of spiral trajectories in low-speed regimes, we introduce and analyze non-spiral algorithms and show that they can \emph{outperform the previously derived spirals} in such settings. This suggests that spiral-search may not always be optimal, as has been shown or conjectured in other models.
We distinguish optimality for the underlying point search problem from optimality within a specified family of trajectories. In particular, we prove optimality within the common-expansion-factor Off-Set family, and the hybrid strategy separates this family from unrestricted multi-speed search.

\subsection{Related Work}
\label{sec: related work}

Search and exploration problems have a long history in operations research, robotics, and theoretical computer science. Early formal models date back to Koopman~\cite{koopman1946search}, with foundational work on deterministic and probabilistic search by Beck and Bellman~\cite{Beck64,Bellman63}. Stone~\cite{stone1975theory} later introduced Bayesian models. This early work evolved into modern search theory, summarized in various surveys and books~\cite{alpern2013search,alpern2006theory,chung2011search,CGK19search,gal2010search,kranakis2024survey}.
A core concept nowadays in search theory is the competitive ratio~\cite{borodin2005online}, defined as the worst-case ratio between the cost of an online algorithm (without full information) and that of an optimal offline strategy. This and related frameworks~\cite{Lopez-Ortiz2016} have led to extensive work on both deterministic~\cite{Albers03,AS11,demaine2006online} and randomized or fault-tolerant search~\cite{CILP13,DDKPU13}.

In one dimension, the classic Cow-Path problem asks a single agent to optimally search for a point. A spiral-like, exponentially expanding strategy is optimal~\cite{baeza1988searching}. Numerous variants have been studied, including cost-based~\cite{demaine2006online}, geometric~\cite{CzyzowiczKKNOS17}, and agent-specific models~\cite{coleman2025multimodal,georgiou2024overcoming}, all employing cyclic, monotonic spirals. An exception arises in the case of probabilistically faulty agents on a half-line, where monotone spirals are suboptimal~\cite{BGMP2022pfaulty}.

In two dimensions, a closely related problem is \emph{shoreline search}, introduced by Bellman~\cite{bellman1958dynamic}, where the target is an infinite line. The case with known shoreline distance to the agent's starting position was solved by Isbell~\cite{isbell1957optimal}, while the general “lost-in-a-forest” version (where the shape of the forest is part of the input) remains open~\cite{berzsenyi1995lost,finch2004lost}, 
with specific convex environments and general shapes having also been studied~\cite{gibbs2016bellman,kubel2021approximation}. 
For unknown distances, the best published strategy known for shoreline search is a logarithmic spiral with competitive ratio $13.81$~\cite{baeza1988searching,finch2005searching}. The best published unconditional lower bound is $6.3972$~\cite{baeza1995parallel}, while the stronger bound $12.5385$ for cyclic paths concerns the axis-parallel shoreline problem~\cite{langetepe2012searching}. More recently, two preprints claim substantially stronger results, which include an unconditional lower bound of $12.5937096701\ldots$~\cite{temerev2026lower}, and a computer-assisted proof that the logarithmic spiral is optimal~\cite{temerev2026optimal}.

When the target is a point, which is the focus of this work, a logarithmic spiral achieves a search cost of $17.289$, and has been proven to be optimal~\cite{langetepe2010optimality}. This result resolved a long-standing conjecture that spiral trajectories are best among all possible strategies. Spiral strategies in planar domains have been extensively studied~\cite{finch2005logarithmic,finch2005searching}. Gal~\cite{gal2010search} established optimality within the class of monotone, periodic strategies, and Langetepe~\cite{langetepe2010optimality} later proved unconditional optimality, confirming the conjecture. For known-distance variants, recent work provides exact and average-case optimal results~\cite{ConleyGeorgiou25Inspection}.
Related geometric search models also admit non-spiral optimal trajectories; for example, recent work based on Fermat’s principle establishes optimal average-case strategies for disk inspection in a different setting~\cite{GeorgiouFermatSTACS26}.

Collaborative search with multiple agents has also received significant attention,
see \cite{Albers00,CFMS10,CILP13} and \cite{DFKNS07,DKP91,DDKPU13,HKLT13,lopez2015optimal}.
The well-known ANTS problem~\cite{feinerman2012collaborative} considers identical agents searching for an adversarially placed treasure in the plane, with performance depending on agent count and target distance. This extends earlier distance-cost models~\cite{baeza1993searching} and connects to other frameworks such as layered graph search~\cite{papadimitriou1991shortest} and ray-based exploration~\cite{angelopoulos2011multi}.
When two-dimensional search involves geometric specifications, the problem becomes more complex. Early work considered polygons~\cite{FeketeGK10} and disks~\cite{CzyzowiczGGKMP14}, later extending to triangles, squares, and $\ell_p$ balls~\cite{BagheriNO19,CzyzowiczKKNOS15,georgiou2022triangle,GLLKllp2023}. Multi-agent extensions introduce coordination and fault-tolerance challenges~\cite{behrouz2023byzantine,BGMP2022pfaulty,CzyzowiczGGKKRW17}, with some studies exploring worst- and average-case trade-offs~\cite{chuangpishit2020multi}.

In multi-agent shoreline search, spiral strategies can improve the search performance. A double spiral achieves a competitive ratio of $5.2644$ for $n=2$~\cite{baeza1995parallel}, and for $n \geq 3$, a ray-based strategy achieves $1/\cos(\pi/n)$ with matching lower bounds~\cite{AcharjeeGKS19,dobrev2020improved}. Partial-information variants (e.g., known slope or orientation) yield diverse bounds~\cite{gluss1961alternative,gluss1961minimax,jez2009two}. Other models consider circular targets~\cite{gluss1961minimax}, memory-limited agents on grids~\cite{emek2015many,Emekicalp2014,fricke2016distributed,LangnerKUW15}, and navigation with cost-information trade-offs~\cite{bouchard2018deterministic,pelc2018reaching,pelc2018information,pelc2019cost}.

Heterogeneous agent capabilities, such as communication~\cite{georgiou2022evacuation} and speed~\cite{BampasCGIKKP19}, introduce further complexity. Of particular relevance to this work is heterogeneity in agent speeds, which has been rarely studied. Related results include multi-agent evacuation~\cite{lamprou2016fast}, energy-efficient sensing~\cite{WIFBZP11}, gathering~\cite{beauquier2010utilizing}, patrolling~\cite{CGKK11,KK12}, and one-dimensional point search in bounded domains~\cite{czyzowicz2014beachcombers}.

\section{Definitions and Main Results}

\subsection{Problem Definition, Notation \& Terminology}

We write $\reals_*^2:=\reals^2\setminus\{O\}$ for the Euclidean plane excluding the \emph{origin} $O=(0,0)$. 
A point in the plane, written in polar coordinates, is denoted by $\langle \rho, \chi \rangle$ (to differentiate from Cartesian coordinates), where $\rho$ is the radial coordinate and $\chi$ the angular coordinate (in radians). 
Thus, $\langle \rho, \chi \rangle = \rho (\coss{\chi},\sinn{\chi})$, and in particular, $O = \langle 0, 0 \rangle$. 
We write $\bone$ to denote a vector of all-$1$s of dimension that will be self-evident from the context. 
For $\psi \in [0,2\pi)$, we define the \emph{ray} $\mathcal L_\psi$ as the set of positive multiples of the vector $\langle 1, \psi \rangle$, that is 
$$
\mathcal L_\psi := \{ \langle x, \psi \rangle : x \in \reals_{> 0} \}.
$$
For a \emph{target} point $A \in \reals^2_*$, we denote by $\norm{A}$ its Euclidean norm, that is its distance from the origin, which is always positive. For $x,y>0$, we say that a point $P=\langle x,\psi\rangle \in \reals_*^2$ \emph{exposes} $A=\langle y,\psi\rangle \in \reals_*^2$ if $x\geq y$, meaning $A$ is a convex combination of $O$ and $P$. 
We use the abbreviation $[n] = \{0,1,2,\ldots,n-1\}$ and refer to $i \in [n]$ as \emph{mobile agent} $i$ (or simply agent-$i$). By convention, agent-$n$ refers to agent-$0$.

We define the problem under consideration as the \emph{Multi-Speed Point Search on $\mathcal D$ by $n$ Agents}, denoted by $\mathcal D\text{-}\msp_n(c)$ (or simply $\msp_n(c)$ when $\mathcal D$ is clear from context), where $c \in \reals_{>0}^n$ is the input and $\mathcal D \subseteq \reals^2_*$ is the \emph{search domain}. 
The parameter $c_i$ represents the speed of mobile agent $i$, with $c_n = c_0$ being the speed of agent-$0$. We assume speeds are given in non-increasing order, i.e., $c_0 \geq c_1 \geq \dots \geq c_{n-1}$, to simplify expressions in our results. Without loss of generality (after appropriate scaling), we assume $c_0 = 1$ is the highest speed. 
We write $\mathcal{G}_l$ for the geometric mean of the speeds of agents $0$ through $l-1$, that is, 
$$\mathcal{G}_l := \left( \prod\nolimits_{i \in [l]} c_i \right)^{1/l}.$$
Observe that $c_l \leq \mathcal G_l$ for every $l=1,\ldots,n-1$.

A \emph{feasible solution} to $\mathcal D\text{-}\msp_n(c)$ consists of continuous trajectories $\{\tau_i\}_{i \in [n]}$ in $\reals^2$, where each agent $i \in [n]$ follows a trajectory $\tau_i$ that depends on both $c$ and $\mathcal D$, satisfying:  \\
- All agents start at the origin.  \\
- Agent $i$ moves along $\tau_i$ at a speed of at most $c_i$.  \\
- Each target $A \in \mathcal D$ is (eventually) exposed by at least one agent. \\
For simplicity, we say that agent $i$ \emph{exposes} a point $A \in \mathcal D$ if some point on $\tau_i$ exposes $A$.  
For each $A \in \mathcal D$, let $t_A$ be the earliest time some agent exposes $A$. The goal is to design trajectories that minimize the \emph{search cost}, defined as  
$$
\sup_{A \in \mathcal D} \frac{t_A}{\norm{A}},
$$
following the normalization of competitive analysis.\footnote{
The search cost corresponds to the notion of competitive ratio in competitive analysis. 
In some settings, the competitive ratio is defined with an additive term, or with an explicit restriction that excludes instances in which the reference cost becomes arbitrarily small (in our case, targets at very small distance $\norm{A}$ from the origin), as a way to deal with pathological inputs. 
In geometric search, efficiency is instead measured directly by the ratio defining the search cost, as we do in this work. 
Excluding the origin from the search domain makes this ratio well defined without imposing an arbitrary positive lower bound on the target distance. 
We use the purely multiplicative search cost throughout, without an additive term, as this is the performance measure considered in this model.
}
Note that the origin is excluded from the search domain $\mathcal D$ so as to have a well-defined search cost. 
Intuitively, the denominator $\norm{A}$ corresponds to the time required by an agent that knows the target location in advance to reach it, while the normalization yields a meaningful worst case measure by comparing the search time as a function of the distance of the hidden target from the origin, which would otherwise be unbounded.

In this work we consider the following search domains $\mathcal D$: \\
- The \emph{Plane Excluding the Origin} defined as 
    $
    \mathcal{P} := \bigcup_{\psi \in [0,2\pi)} \mathcal{L}_\psi = \mathbb{R}^2_*
    $, \\
- the \emph{Cone of Angle $\phi$} 
defined as
    $
    \mathcal{C}_\phi := \bigcup_{\psi \in [0,\phi)} \mathcal{L}_\psi
    $, 
where $\phi \in (0, \pi)$, and \\
- the \emph{Conic Complement of Angle $\phi$} defined as 
    $
    \mathcal{W}_\phi := \mathcal{P} \setminus \mathcal{C}_\phi
    $, where $\phi \in (0, \pi)$.

The notion of exposure corresponds to a well-known search/visibility criterion; a mobile agent sees a point if and only if the point lies between the agent's location and the origin. 
Agents are modelled as points and do not interfere with one another; in particular, the exposure condition is defined per agent and is not affected by the locations of other agents (there is no notion of occlusion or blocking in the model).
Our problem is to design trajectories for $n$ agents, considering their speed constraints, to locate a hidden target while minimizing the worst-case time required to expose a point relative to its distance from the origin.  

Since trajectories are only relevant until a target is exposed, every participating agent may be assumed to move at full speed; agents that are not used by an algorithm may remain at the origin. 
Thus, we model agent trajectories as locally absolutely continuous functions (in particular, the continuous piecewise smooth curves used below)
$
\tau_i:\reals\to\reals^2,
$
where each $\tau_i$ satisfies
$
\lim_{t\to -\infty}\tau_i(t)=(0,0),
$
with the understanding that $t$ is an abstract parameter and does not represent physical time.

The above convention should be understood as a normalization. 
The parameter $t$ is a curve parameter, and the actual time is obtained from arc length. 
The curve parameter ranges over all real values, with the trajectory approaching the origin as $t\to-\infty$ (the execution begins at time $0$).
In particular, this avoids introducing an arbitrary positive lower bound on the radius (or an arbitrary starting offset along the curve) when describing spiral-type trajectories.
We emphasize that the agents do not physically start on a common positive-radius circle. Choosing a reference parameter at which a trajectory meets such a circle merely reparametrizes the same curve (it does not replace the common physical starting point at the origin). In the multi-speed setting, imposing equal positive radii at a common positive time would also be artificial, since equal radii do not correspond to equal travel times for agents of different speeds.

Local absolute continuity ensures that the derivative exists almost everywhere and that arc length on each finite parameter interval is the integral of its norm. We additionally require finite initial arc length at every finite parameter value. A finite-length trajectory is parametrized to include its endpoint and to be constant thereafter.

Since the search domains considered here are unbounded, every feasible collection contains trajectories of unbounded total length. Feasibility itself is the separate requirement that, for each $A \in \mathcal D$, there exist $i \in [n]$, $s_A \in \reals$, and $\lambda \in [0,1]$ such that
$
A = \lambda \tau_i(s_A).
$

Distances are measured with respect to the standard Euclidean metric on $\reals^2$, and speed bounds are interpreted with respect to arc length in this metric.
For every $s\in\reals$, define the curve length from the origin to $\tau_i(s)$ by
$$
\ell_i(s) :=
\int_{-\infty}^{s}
\norm{
\tau_{i}'(t)
}
\dd t.
$$
We require $\ell_i(s)<\infty$ for every $s\in\reals$. Then $\ell_i(s)\to0$ as $s\to-\infty$, and arc-length parametrization at speed $c_i$ gives a continuous physical trajectory starting at time zero. For any unused agent we set $\tau_i(s)=O$ and $\ell_i(s)=0$ for all $s$.

Since agent-$i$ is assumed to have speed $c_i$, the physical time required for agent-$i$ to reach $\tau_i(s_A)$ is
$
\frac{\ell_i(s_A)}{c_i}.
$

For each agent $i$ and target $A\in\mathcal D$, define its discovery time by
$$
T_i(A):=\inf\left\{\frac{\ell_i(t)}{c_i}:t\in\reals,\ A\in[O,\tau_i(t)]\right\},
$$
and we also set $\inf\varnothing:=\infty$. 
Thus, an agent that never exposes $A$ contributes infinity, independently of its total trajectory length.
The first-discovery time is $t_A=\min_{i\in[n]}T_i(A)$, and the objective is
$$
\sup_{A\in\mathcal D}\frac{1}{\norm A}\min_{i\in[n]}T_i(A).
$$
We say that agent $i$ \emph{first exposes} $A$ when $T_i(A)=t_A<\infty$; this allows ties unless a unique first-exposing agent is specified.

\subsection{Main Contributions \& Paper Organization}
\label{sec: main contributions}

In this section, we summarize our main contributions, focusing on a simplified presentation of our results without explicitly describing the search algorithms. 
The primary focus of our work is to establish upper bounds to the search cost of $\mathcal{P}\text{-}\msp_n(c)$, i.e. the problem of searching the plane with arbitrary speed agents. As an essential component for the follow-up results, we first consider the simpler case $\mathcal{P}\text{-}\msp_n(\bone)$, where all $n$ agents have unit speed. The proof of the following upper bound, based on a logarithmic-spiral type trajectory, generalizes the optimal solution for $\mathcal{P}\text{-}\msp_1(1)$ from \cite{langetepe2010optimality}. The notation introduced in the following theorem, as well as the derived results, are essential for presenting our more general results.

\begin{theorem}
\label{thm: spiral with n unit speed agents}
For every $n \in \naturals$, the search problem $\mathcal{P}\text{-}\msp_n(\bone)$ admits a solution with search cost  
$$
\mathcal{U}_n := \tfrac{\sqrt{1+\kappa_n^2}}{\kappa_n}e^{2\pi \kappa_n/n},
$$
where $\kappa_n$ admits a closed-form expression, obtained as the root of an explicit cubic polynomial (see Lemma~\ref{lem: closed form of k_i}).
\end{theorem}
The proof of Theorem~\ref{thm: spiral with n unit speed agents} is given in Section~\ref{sec: n unit speed agents} and is a simple generalization of the single-agent search analysis.
Some values for $\mathcal U_n$ are given in Table~\ref{tab:transposed} on page~\pageref{tab:transposed}. 
Quite importantly, we also show that $\mathcal{U}_n$ is decreasing in $n$ and approaches $1$ as $n \to \infty$. 

This motivates the question of how the performance of $n$ multi-speed agents compares to that of $n$ idealized unit-speed agents in plane search. The following theorem provides one of our \emph{main technical contributions} by establishing a general upper bound, expressed in terms of the previously established quantities $\mathcal{U}_i$ and $\kappa_i$.
We remind the reader that $\mathcal{G}_\ell$ denotes the geometric mean of the speeds of the fastest $\ell$ agents $[\ell]$. 
For convenience, we set $\mathcal{G}_0 = 1$ and define the search cost $\mathcal{U}_0$ for zero agents as positive infinity.
\begin{theorem}
\label{thm: simplified general spiral upper bound}
$\mathcal{P}\text{-}\msp_n(c)$ admits a solution with search cost
$$
\min_{i=1,\ldots,n} \frac{\mathcal{U}_i}{\mathcal{G}_i}.
$$
More precisely, let
$$
\nu \in \argmin_{i=1,\ldots,n} \frac{\mathcal{U}_i}{\mathcal{G}_i}
$$
and set $\ell=\nu-1$. The bound is achieved by using only the agents $[\nu]=\{0,1,\ldots,\ell\}$. If $\ell\geq 1$, then the slowest participating agent necessarily satisfies
$$
c_\ell \geq \left(\mathcal{U}_{\ell+1}/\mathcal{U}_{\ell}\right)^{\ell+1}\mathcal{G}_\ell.
$$
This is also the minimum cost within the common-expansion-factor Off-Set family of Algorithm~\ref{algo: Off-Set Trajectory}, with full-speed motion, arbitrary angular phase assignments, and unused agents allowed.
\end{theorem}
The proof of Theorem~\ref{thm: simplified general spiral upper bound} appears in Section~\ref{sec: arbitrary speeds spiral},
and more specifically as the corollary of a sequence of lemmata, in Section~\ref{sec: simplified general spiral upper bound}. 
The proof constructs logarithmic-spiral trajectories tailored to the agents' speeds.  
Crucially, these trajectories introduce enough degrees of freedom to balance the search cost across the agents that contribute to search.
Note that ordering the speeds in non-increasing order simplifies the theorem's statement, as it maximizes the geometric mean over any subset of agents $[j]$, where $j \leq n$.  
Since $c_\ell \leq \mathcal{G}_\ell$, Theorem~\ref{thm: simplified general spiral upper bound} shows that the slowest agent in any prefix attaining the displayed minimum must be sufficiently fast relative to the geometric mean of the faster agents. The condition is necessary for that prefix to improve on the prefix obtained by omitting its slowest agent; it is not, by itself, sufficient for global optimality among all prefixes.
This threshold is quantified by the factor
$$
\left(\mathcal{U}_{\ell+1}/\mathcal{U}_{\ell}\right)^{\ell+1}
=1-(2\pi)^{2/3}\ell^{-2/3}+O(\ell^{-4/3}).
$$
This follows by a Taylor expansion of the formula of $\mathcal U_\ell$ from Theorem~\ref{thm: spiral with n unit speed agents} and Proposition~\ref{prop: monotonicity of Cn} on page~\pageref{prop: monotonicity of Cn} regarding the behaviour of $\kappa_n$, with representative numerical values provided in Table~\ref{tab: geometric mean factor}.  
The following corollary follows immediately from Theorem~\ref{thm: simplified general spiral upper bound}.
\begin{table}[h]
    \centering
\begin{tabular}{|c|c|c|c|c|c|c|c|c|}
\hline
$\ell$
& 1
& 2
& 3
& 4
& 5
& 10
& 100
& 1000
\\ \hline
$p_\ell$
& 0.26668
& 0.33831
& 0.38602
& 0.42395
& 0.45579
& 0.56406
& 0.86070
& 0.96689
\\ \hline
\end{tabular}
\caption{Some indicative numerical values of 
$p_\ell :=
\left(\mathcal{U}_{\ell+1}/\mathcal{U}_{\ell}\right)^{\ell+1}
$
as it appears in Theorem~\ref{thm: simplified general spiral upper bound}.
}
\label{tab: geometric mean factor}
\end{table}
\begin{corollary}
\label{cor: bound on geometric mean}
Searching for a point in the plane with $n$ agents of arbitrary positive speeds $c \in \mathbb{R}_{>0}^n$ can be performed with cost at most $\mathcal U_k$, the cost of the spiral construction of Theorem~\ref{thm: spiral with n unit speed agents} for $k\leq n$ unit-speed agents, provided that the geometric mean of the $n$ speeds is at least
$
\mathcal U_n/\mathcal U_k.
$
The inequality is strict when the geometric mean strictly exceeds $\mathcal U_n/\mathcal U_k$.
\end{corollary}

Next, we explore alternative search trajectories for searching the plane, which, for certain agent speed ranges, strictly improve the upper bounds of Theorem~\ref{thm: simplified general spiral upper bound} obtained from spiral-type trajectories.  
To this end, we establish upper bounds for searching cones $\mathcal C_\phi$ with $\phi\in(0,\pi/2)$ and conic complements $\mathcal W_\phi$ with $\phi\in(0,\pi)$, which we believe are of independent interest. Our contributions focus on proposing search trajectories and analyzing their competitive performance for single unit-speed agents.  
Due to the technical nature of the results, the derived upper bounds are expressed in terms of the minimum values of certain functions parameterized by the angle $\phi$.  
An analytic comparison of these minima establishes strict improvement over the bound of Theorem~\ref{thm: simplified general spiral upper bound} for $\mathcal{P}\text{-}\msp_2(1,c)$ on a precisely defined speed interval, and with numerical evaluations we locate its endpoint.

We begin by analyzing upper bounds for searching for a point in an angle-$\phi$ cone with a unit-speed agent.  

\begin{theorem}
\label{thm: cone search upper bound}
For $\phi \in (0,\pi/2)$, the problem $\mathcal C_\phi\text{-}\msp_1(1)$ admits a solution with search cost  
$$
\coss{\phi/2} \cdot \mathcal F_\phi,
$$
where $\mathcal F_\phi$ is the minimum of the function  
$$
\left(\frac{\sinn{\theta}}{\sinn{\theta - \phi}}\right)^2 \cdot \frac{1}{\coss{\theta - \phi/2}}
$$
over the interval $\theta \in (\phi , \pi/2]$.  
\end{theorem}

We prove Theorem~\ref{thm: cone search upper bound} in Section~\ref{sec: Cone Search}.  
Next, we derive upper bounds for searching the conic complement of angle $\phi$ with a unit-speed agent, for $\phi \in (0,\pi)$. 
Importantly, the associated search cost will be strictly smaller than $\mathcal U_1$, which is optimal for searching the plane with a unit-speed agent. The following theorem establishes this result.

\begin{theorem}
\label{thm: wedge search upper bound}
For $\phi \in (0,\pi)$ and every $\lambda >1$, define  
$$
k_\lambda
=
\frac{\ln \lambda}{2\pi}
+
\frac{\lambda -1}{\sqrt{2\lambda(1-\coss{\phi}) }}
\left(
1-\frac{\phi}{2\pi}
\right).
$$
Then, the problem $\mathcal W_\phi\text{-}\msp_1(1)$ admits a solution with search cost  
$\mathcal R_\phi$, where $\mathcal R_\phi$ is the minimum of the function
$$
\frac{\sqrt{\lambda ^2-2 \lambda  \cos (\phi )+1}}{\lambda -1}
e^{2\pi k_\lambda}
$$
in the interval $\lambda \in (1,\infty)$. 
Moreover, for every $\phi \in (0,\pi)$ this value is strictly smaller than $\mathcal U_1$.
\end{theorem}

We prove Theorem~\ref{thm: wedge search upper bound} in Section~\ref{sec: Conic Complement Search}.  
Having established these bounds, we now show that the common-expansion-factor Off-Set family is not always optimal for searching the plane with multiple agents.
We focus on $\mathcal P\text{-}\msp_2(1,c)$, where two agents, one unit-speed and one speed-$c$, search the plane. We derive an alternative upper bound to Theorem~\ref{thm: simplified general spiral upper bound}, achieved through a hybrid strategy that combines cone and conic complement search techniques.

\begin{theorem}
\label{thm: hybrid upper bound}
Let  
$$
\gamma(\phi) = \coss{\phi/2} \frac{\mathcal F_\phi}{\mathcal R_\phi},
$$
where $\mathcal F_\phi$ and $\mathcal R_\phi$ are as defined in Theorems~\ref{thm: cone search upper bound} and~\ref{thm: wedge search upper bound}, respectively.  
Then, for every $\phi\in(0,\pi/2)$ satisfying $\gamma(\phi)\leq1$, the problem $\mathcal P\text{-}\msp_2(1,\gamma(\phi))$ admits a solution with search cost equal to $\mathcal R_\phi$, or also equal to 
$$
\coss{\phi/2} \frac{\mathcal F_\phi}{\gamma(\phi)}.
$$
\end{theorem}

We prove Theorem~\ref{thm: hybrid upper bound} in Section~\ref{sec: suboptimality}.  
At a high level, we use Theorem~\ref{thm: wedge search upper bound} to assign the unit-speed agent to search a conic complement, ensuring a search cost strictly smaller than $\mathcal U_1$.  
We then assign the speed-$c$ agent to search the corresponding cone using the trajectory from Theorem~\ref{thm: cone search upper bound}, carefully tuning the cone angle $\phi=\phi(c)$ so that its search cost does not exceed that of the faster agent.  
When the cone is sufficiently narrow, i.e. $\phi$ is small enough, even a slow agent can explore it efficiently.  
The relationship between the cone angle $\phi$ and the required agent speed is captured by the function $\gamma(\phi)$, which defines the minimum speed needed to ensure balanced competitive costs.

Notably, from our analysis, Theorem~\ref{thm: simplified general spiral upper bound} guarantees an upper bound of $\mathcal U_1$ for 
$$c \leq (\mathcal U_2/\mathcal U_1)^2 \approx 0.2666865601$$ 
when searching the plane with two agents (see also Table~\ref{tab: geometric mean factor}), one unit-speed and one with speed $c$. This matches the bound achieved by a single unit-speed agent
(in a minimizing Off-Set construction, the speed-$c$ agent may therefore be ignored in this range). 
At $c=(\mathcal U_2/\mathcal U_1)^2$, both the one-agent and two-agent choices attain the same bound.
Furthermore, we show that if 
$$
c \leq 1/\mathcal U_1 \approx 0.05783906146,
$$
then $\mathcal U_1$ is the optimal search cost for $\mathcal P\text{-}\msp_2(1,c)$. Theorem~\ref{thm: hybrid upper bound} improves the search cost of $\mathcal U_1$, when $c > 1/\mathcal U_1$ and for speeds slightly exceeding $0.2666865601$. This threshold corresponds to the speed at which both agents, with speeds $1$ and $c$, expose new targets in the spiral-type trajectory used in Theorem~\ref{thm: simplified general spiral upper bound}.  
The next theorem proves that the common-expansion-factor Off-Set family is suboptimal for an interval of low speeds. Its proof establishes a unique transition value analytically, while numerical evaluation gives the reported transition value. We should emphasize though that the comparison does not rule out other families of spiral trajectories.
Since Theorem~\ref{thm: simplified general spiral upper bound} already yields strong upper bounds for larger values of $c$, we focus below on the case of small speeds.
\begin{theorem}
\label{thm: numerical upper bound}
For $c\in(0,0.5]$, set $C(\phi)=\cos(\phi/2)\mathcal F_\phi$, and let $\phi_*$ be the unique solution in $(0,\pi/2)$ of
$
C(\phi_*)\mathcal R_{\phi_*}=\mathcal U_2^2,
$
and set
$
c_*:=\gamma(\phi_*)
$
(numerically,
$
\phi_*=0.7510108432\ldots
$
and
$
c_*=0.2681784154\ldots
$).
For $1/\mathcal U_1<c<c_*$, let $\phi(c)\in(0,\pi/2)$ be the unique angle satisfying $\gamma(\phi(c))=c$. Then, $\mathcal P\text{-}\msp_2(1,c)$ admits the upper bound
$$
\begin{cases}
\mathcal U_1, & 0<c\leq 1/\mathcal U_1,\\[2mm]
\mathcal R_{\phi(c)}, & 1/\mathcal U_1<c<c_*,\\[2mm]
\mathcal U_2/\sqrt{c}, & c_*\leq c\leq0.5.
\end{cases}
$$
For $1/\mathcal U_1<c<c_*$, this bound is strictly better than the minimum Off-Set cost and is achieved by the Cone-Wedge Hybrid Algorithm. For $c\geq c_*$, it is achieved by the Off-Set construction of Theorem~\ref{thm: simplified general spiral upper bound}, while at $c=c_*$ the two constructions have the same cost. For $0<c\leq1/\mathcal U_1$, the bound is achieved by the unit-speed agent alone.
\end{theorem}
 Figure~\ref{fig: num upper bound plane 2 agents} depicts the upper bound of Theorem~\ref{thm: numerical upper bound}. 
\begin{figure}[h!]
    \centering
    \includegraphics[width=0.42\textwidth]{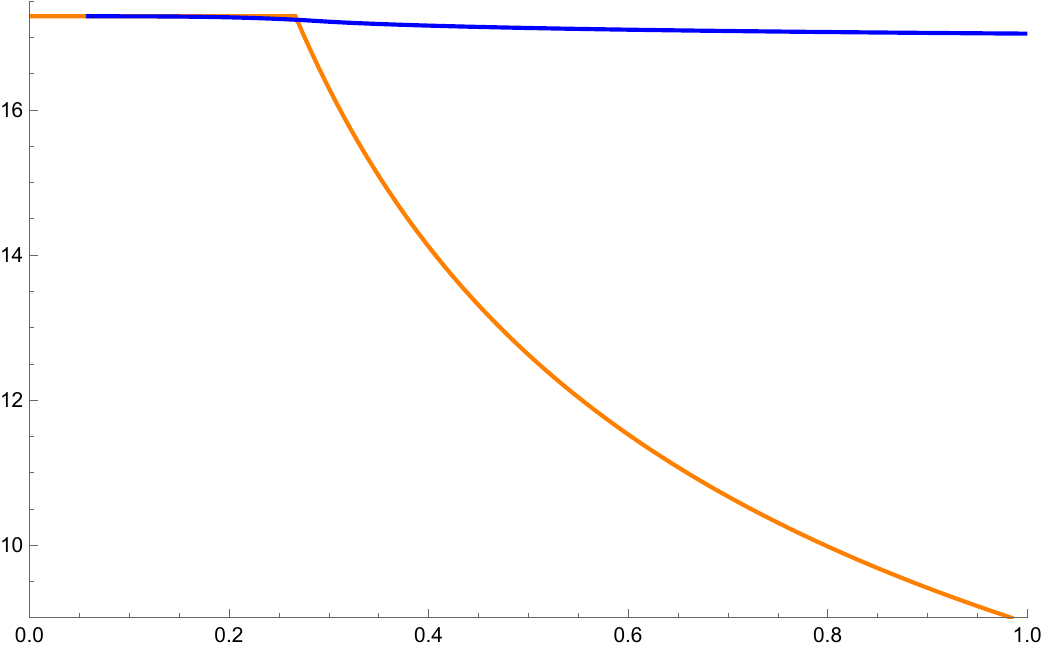}
    \includegraphics[width=0.42\textwidth]{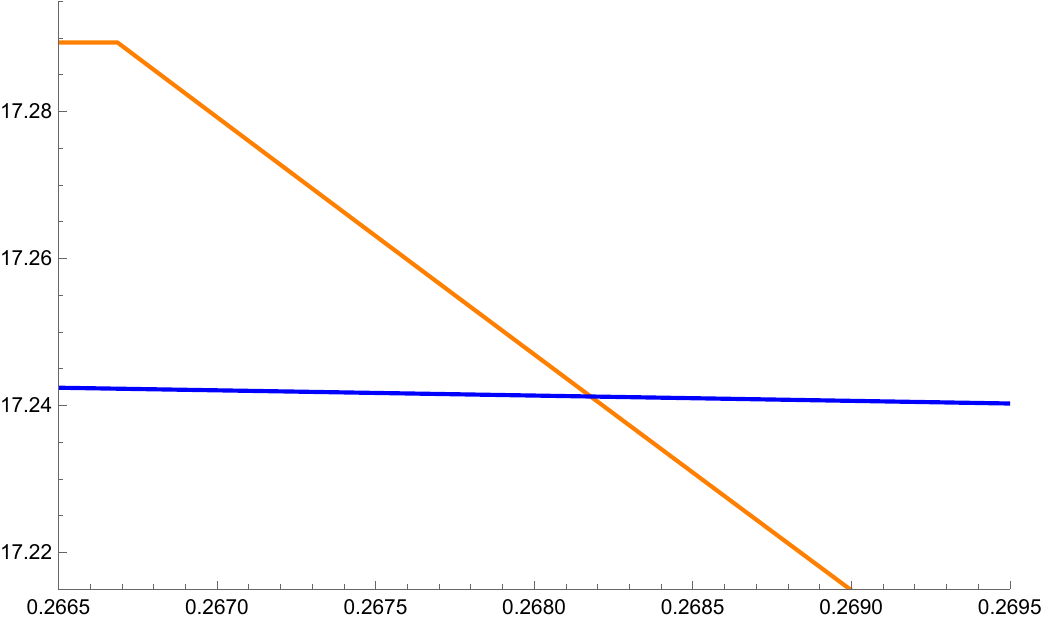}
    \caption{
The orange curve represents the upper bound from Theorem~\ref{thm: simplified general spiral upper bound}, achieved by the Off-Set Algorithm~\ref{algo: Off-Set Trajectory} for searching the plane with agent speeds $1,c$, as a function of $c \in (0,1)$. 
The transition value where the Off-Set upper bound is not differentiable in the speed is $c=(\mathcal U_2/\mathcal U_1)^2\approx0.2666865601$, and note that past this threshold speed value, this bound is strictly lower than $\mathcal U_1$.
The blue curve represents the upper bound from Theorem~\ref{thm: hybrid upper bound}, achieved by the Cone-Wedge Hybrid Algorithm~\ref{algo: cone wedge hybrid} for searching the plane with agent speeds $1,c$. The right panel zooms in on the region where the two bounds converge, showing that the latter strategy is strictly better for searching the plane  even for $c \in (0.2666865601,0.2681784154)$ in which both the speed-$1$ and speed-$c$ agents expose new targets following the Off-Set Algorithm~\ref{algo: Off-Set Trajectory}.
    }
    \label{fig: num upper bound plane 2 agents}
\end{figure}
The hybrid beats the two-active-agent Off-Set bound on the interval $((\mathcal U_2/\mathcal U_1)^2,c_*)$. Below the activation threshold, the hybrid also improves on the one-agent bound whenever $c>1/\mathcal U_1$. For $c\leq1/\mathcal U_1$, the slow agent cannot reduce the optimal cost (for $c\geq c_*$, the Off-Set bound is no larger than the balanced hybrid bound). The intermediate regime visible in Figure~\ref{fig: num upper bound plane 2 agents} reflects the balance between searching a narrow cone and its complementary region, established analytically in Section~\ref{sec: suboptimality}.
We prove Theorem~\ref{thm: numerical upper bound} in Section~\ref{sec: suboptimality}.

\section{Searching the Plane with $n$ Unit-Speed Agents}
\label{sec: n unit speed agents}
In this section, we prove Theorem~\ref{thm: spiral with n unit speed agents}, establishing upper bounds for the problem $\mathcal P\text{-}\msp_n(\bone)$, where $n$ unit-speed agents search the plane. The derived upper bounds naturally generalize the provably optimal algorithm for $\msp_1(1)$ due to~\cite{langetepe2010optimality}.  
For a single unit-speed agent, an optimal trajectory follows a logarithmic spiral, parameterized by the \emph{expansion factor} $k > 0$. 
The trajectories for $n$ unit-speed agents are based on the same principle and are defined as follows.

\begin{figure}[h!]
    \centering
    \includegraphics[width=0.47\textwidth]{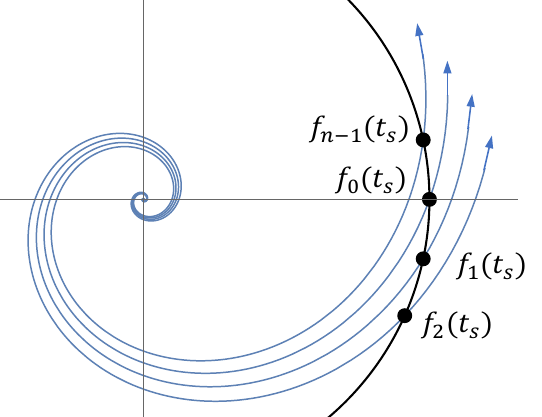}
    \caption{
 The depiction of trajectories of Algorithm~\ref{algo: uniform spiral trajectory}. 
 In this example, we use $k=0.356$, $n=30$, and we show the trajectories of agents $0,1,2,29$, along with their positions when $t=t_s = 2\pi s$, for some $s\in \integers$. 
 All agents lie on the circle of radius $e^{k t_s}$, while agent-$0$ lies on ray $\mathcal L_0$. The last agent who previously hit $\mathcal L_0$ was agent-$(n-1)$. 
    }
    \label{fig: uniform spiral}
\end{figure}

\begin{algorithm}[H]
\caption{Uniform Spiral Trajectory for $\mathcal P\text{-}\msp_n(\bone)$}
\label{algo: uniform spiral trajectory}
\begin{algorithmic}
\REQUIRE $k > 0$
\STATE \textbf{Output Trajectories:}  
$
f_i^k(t) = \langle 
e^{kt}, 
t - \tfrac{2\pi}{n}i
\rangle,
\hspace{0.5em} t \in \reals, \hspace{0.5em} i \in [n].
$
\end{algorithmic}
\end{algorithm}
The logarithmic spiral guarantees that the distance of newly exposed points from the origin grows proportionally to the length of the trajectory, which yields a constant search cost. 
Using the same expansion factor across agents further guarantees that trajectories do not overlap.

Algorithm~\ref{algo: uniform spiral trajectory} is depicted in Figure~\ref{fig: uniform spiral}. 
The execution of the algorithm begins at $t \to -\infty$.  
The condition $k > 0$ is necessary for correctness, ensuring that every point in the plane is eventually exposed. 
For any fixed parameter value $t$, all agents have the same distance $e^{kt}$ from the origin, while their angular separation is exactly $2\pi/n$. Consequently, at every snapshot of the algorithm, all agents lie on a common circle and their positions form a regular $n$-gon (up to rotation and scaling as $t$ varies).

The following lemma establishes that Algorithm~\ref{algo: uniform spiral trajectory} is indeed a feasible search strategy.  

\begin{lemma}
\label{lem: feasible of uniform spiral}
For every $k > 0$ and $n \in \naturals$, Algorithm~\ref{algo: uniform spiral trajectory} is feasible for $\mathcal P\text{-}\msp_n(\bone)$.
\end{lemma}

\begin{proof}
For every $i \in [n]$ and $k > 0$, we observe that  
$$
\lim_{t\rightarrow -\infty} \norm{ f_i^k(t) } = \lim_{t\rightarrow -\infty} e^{kt} = 0,
$$  
which confirms that agents start from the origin. 
By Lemma~\ref{lem: standard integral} below, the initial arc length to parameter $t$ is $\sqrt{1+k^2}e^{kt}/k<\infty$, so the curves admit the required unit-speed execution from time $0$.

Next, we show that every point $A = \langle x, \psi \rangle$, with $x > 0$ and $\psi \in [0,2\pi)$, is exposed. Consider the smallest $t$ such that $e^{kt} \geq x$, which exists since $k > 0$. Define the smallest $s \in \integers$ such that $2\pi s + \psi \geq t$. Then,  
$
f_0^k(2\pi s + \psi) \in \mathcal L_\psi,
$  
and  
$
\norm{ f_0^k(2\pi s + \psi) } = e^{k(2\pi s + \psi)} \geq e^{kt} \geq x,
$  
meaning that agent-$0$ exposes $A$.  
\qed \end{proof}

The parameter $t$ in Algorithm~\ref{algo: uniform spiral trajectory} is related to, but not equal to, time when unit-speed agents follow trajectories $f_i^k(t)$. The next lemma explicitly describes this relationship and shows that the trajectory lengths $f_i^k(t)$ are invariant with respect to $i$.

\begin{lemma}
\label{lem: standard integral}
Fix $\phi\in \reals$, $\alpha,\beta\in \reals$ with $\alpha > 0$, and consider $h:\reals\mapsto \reals^2$ with 
$h(t) = \langle
 e^{\alpha t + \beta},
 t+\phi
 \rangle$. 
 Then 
$$
\int_{-\infty}^{y} \norm{ h'(t) } dt
=
 \frac{\sqrt{1+\alpha^2}}{\alpha} e^{\alpha y + \beta}. 
$$
\end{lemma}

\begin{proof}
In Cartesian coordinates, we have
$h(t) = e^{\alpha t + \beta} \left(
\coss{t+\phi }, \sinn{t+\phi}
\right)$.
Therefore, 
$$
\dd h(t) /\dd t =
e^{\alpha t +\beta} \left(  \alpha \cos (t+\phi )-\sin (t+\phi ),  \alpha \sin (t+\phi )+\cos (t+\phi )\right),
$$
and hence 
\begin{align*}
\norm{ \dd h(t) /\dd t }^2 
&= e^{2(\alpha t +\beta)} 
\left(
\left(  \alpha \cos (t+\phi )-\sin (t+\phi ) \right)^2 + \left(  \alpha \sin (t+\phi )+\cos (t+\phi )\right)^2 
\right)\\
& = e^{2(\alpha t + \beta)}  \left( 1+\alpha^2 \right).
\end{align*}
Therefore, 
$$
\int_{-\infty}^{y}
\norm{
\dd h(t) /\dd t
}
\dd t
= 
\sqrt{1+\alpha^2}
\int_{-\infty}^{y}
e^{\alpha t + \beta}
\dd t
= 
\frac{\sqrt{1+\alpha^2}}{\alpha} e^{\alpha y + \beta},
$$
where the last equality is due to the fact that $\alpha>0$. 
\qed \end{proof}

As an immediate corollary, we have that for the trajectories of Algorithm~\ref{algo: uniform spiral trajectory}, and for each $i \in [n]$, 
$
\int_{-\infty}^{y} \norm{ (f_i^k)'(t) } dt
=
 \frac{\sqrt{1+k^2}}{k} e^{k y}. 
$
We can now derive the following upper bound, which for $n=1$, and a proper choice of $k$, gives the known optimal result for $\mathcal P\text{-}\msp_1(1)$ due to~\cite{langetepe2010optimality}.

\begin{lemma}
\label{lem: unit speed spiral bounds}
The search cost of Algorithm~\ref{algo: uniform spiral trajectory} and input $k>0$ for $\mathcal P\text{-}\msp_n(\bone)$ equals
$$
\frac{\sqrt{1+k^2}}{k} e^{2\pi k/n}.
$$
\end{lemma}

\begin{proof}
Let $R_\eta$ denote rotation by angle $\eta$. Observe that
$$
e^{-k\psi}R_{-\psi}f_i^k(t)=f_i^k(t-\psi).
$$
Hence, after rotating the plane by $-\psi$ and scaling distances by $e^{-k\psi}$, a target on $\mathcal L_\psi$ is mapped to $\mathcal L_0$. Both the travel time and the target distance are scaled by the same factor, so the search cost remains unchanged.

Also, rotation by $2\pi i/n$ permutes the uniform trajectories and maps agent $i$ to agent $0$. Therefore, without loss of generality, we may assume that the target lies on $\mathcal L_0$ and that it is first exposed by agent $0$.

Define $t_s := 2\pi s$ for $s \in \integers$. That is, the target $A_x = (x,0)$ is first exposed when agent-$0$ reaches $f_0^k(t_s)$ for some $s \in \integers$. It follows that $x \leq e^{kt_s} = e^{k(2\pi s)}$. Moreover, by Lemma~\ref{lem: standard integral}, agent-$0$ reaches this point at time  
$$
\ell_0(t_s) = \frac{\sqrt{1+k^2}}{k} e^{k t_s} = \frac{\sqrt{1+k^2}}{k} e^{k(2\pi s)}.
$$

Next, we establish a lower bound on $x$ under the assumption that agent-$0$ is the first to expose $A_x$. Since time increases with the parameter $t$ along the trajectories, the largest $t < t_s$ for which $f_j^k(t)$ was on the positive horizontal axis is given by  
$$
t + \frac{2\pi(n-j)}{n} \leq t_s = 2\pi s,
$$  
which implies  
$$
t \leq 2\pi s - \frac{2\pi(n-j)}{n} \leq 2\pi s - \frac{2\pi}{n},
$$  
where the last inequality is tight for $j = n-1$, see also Figure~\ref{fig: uniform spiral}. 
Thus, before $A_x$ was exposed, the last agent to cross the positive horizontal axis was agent-$(n-1)$, reaching the point $f_{n-1}^k(2\pi s - 2\pi/n)$ and exposing all points in the closed interval $\left[0, e^{k(2\pi s - 2\pi/n)}\right]$. It follows that  
$
x > e^{k(2\pi s - 2\pi/n)},
$  
with strict inequality.  
Hence, the worst-case search cost for points exposed by agent-$0$ is  
$$
\sup_{s\in \integers}
\sup_{x \in \left( e^{k(2\pi s - 2\pi/n)} , e^{k(2\pi s)} \right] } 
\frac{\ell_0(t_s)}{x}
=
\sup_{s\in \integers}
\frac{\sqrt{1+k^2}}{k} 
\frac{e^{k (2\pi s)}}
{e^{k(2\pi s - 2\pi/n)}}
=
\frac{\sqrt{1+k^2}}{k} e^{2\pi k/n}.
$$
In particular, these calculations show that the search cost is independent of $s$ in the above analysis and also independent of agent-$i$, as established by the symmetry argument at the beginning of the proof.
\qed \end{proof}

Minimizing the search cost expression in Lemma~\ref{lem: unit speed spiral bounds} with respect to $k = k(n)$ yields the following upper bounds for $\mathcal P\text{-}\msp_n(\bone)$, along with the corresponding optimal expansion factors, denoted by $\kappa_n$. The next lemma provides a closed formula for $\kappa_n$.
\begin{lemma}
\label{lem: closed form of k_i}
For each $n \in \naturals$, 
the search cost of trajectories $\{f_i^k(t)\}_{i\in [n]}$
for $\mathcal P\text{-}\msp_n(\bone)$, as seen in Lemma~\ref{lem: unit speed spiral bounds}, is convex in expansion factor $k$. Moreover, the corresponding optimizers $k=\kappa_n$ 
are given as the unique positive root to 
$k^3+k=n/(2\pi)$, and therefore
$$
\kappa_n :=
 \sqrt[3]{\sqrt{\frac{n^2}{16 \pi^2}+\frac{1}{27}}+\frac{n}{4 \pi}}
-\sqrt[3]{\sqrt{\frac{n^2}{16 \pi^2}+\frac{1}{27}}-\frac{n}{4 \pi}}.
$$
\end{lemma}

\begin{proof}
Consider the function $z(k) = \frac{\sqrt{1+k^2}}{k} e^{k/a}$ parameterized by $a \in \reals_{>0}$. We compute  
$$
\frac{\dd}{\dd k} z(k) = \frac{e^{k/a} \left(-a+k^3+k\right)}{a k^2 \sqrt{k^2+1}},
$$  
so that any $k_a$ satisfying $k_a^3 + k_a = a$ is a critical point. We show that this corresponds to a global minimum.  
Indeed, computing the second derivative, we obtain  
$$
\tfrac{\dd^2}{\dd k^2} z(k) =
e^{k/a}
\frac{  a^2 \left(3 k^2+2\right) +\left(k^3+k\right)\left(k^3+k-2a\right) 
}
{a^2 k^3 \left(k^2+1\right)^{3/2}}.
$$  
It follows that  
$$
\frac{\dd^2}{\dd k^2} z(k_a) = \frac{\left(3 k_a^2+1\right) e^{k_a/a}}{k_a^3 \left(k_a^2+1\right)^{3/2}} >0,
$$  
confirming that $k_a$ is indeed a minimizer.  

More generally, we analyze the discriminant of the numerator in  $\tfrac{\dd^2}{\dd k^2} z(k)$, which equals
$$
a^2 \left(3 k^2+2\right) +\left(k^3+k\right)\left(k^3+k-2a\right),
$$  
and we view it as a degree-2 polynomial in $a$. The discriminant evaluates to  
$$
-4k^2(1+k^2)^2(1+3k^2)<0.
$$  
Since this is negative, the polynomial maintains a constant sign, which is determined by the leading coefficient, $3k^2+2>0$. It follows that $\frac{\dd^2}{\dd k^2} z(k) > 0$ for all $k$, proving that $z(k)$ is convex.  

Next, we derive a closed-form expression for $k_a$, which satisfies the equation  
$
k^3 + k - a = 0.
$  
The discriminant of this degree-3 polynomial is  
$
-4(1)^3 - 27(-a)^2 = -4 - 27a^2 < 0.
$  
Since the discriminant is negative, the equation has only one real root, given by Cardano’s formula, and that is
$$
k_a =
\left( \sqrt{\frac{a^2}{4} + \frac{1}{27}} + \frac{a}{2} \right)^{\frac{1}{3}}
- \left( \sqrt{\frac{a^2}{4} + \frac{1}{27}} - \frac{a}{2} \right)^{\frac{1}{3}}.
$$  
Setting $\kappa_n = k_a$ with $a=n/(2\pi)$ completes the proof.
\qed \end{proof}

The proof of Theorem~\ref{thm: spiral with n unit speed agents} is obtained now as an easy corollary. 

\begin{proof}[of Theorem~\ref{thm: spiral with n unit speed agents}]
We use the trajectories from Algorithm~\ref{algo: uniform spiral trajectory} with expansion factor $\kappa_n$ from Lemma~\ref{lem: closed form of k_i}. By Lemma~\ref{lem: unit speed spiral bounds}, the resulting search cost is as stated in the theorem.
\qed \end{proof}

\begin{table}[h]
    \centering
    \begin{tabular}{|c|c|c|c|c|c|c|c|}
        \hline
        $n$ & 1 & 2 & 3 & 4 & 5 & 6 & 7  \\
        \hline
        Search Cost $\mathcal U_n$ & 17.28935 & 8.92852 & 6.22135 & 4.90385 & 4.13049 & 3.62351 & 3.2659   \\
        \hline
        $\kappa_n$ & 0.155402 & 0.293124 & 0.409031 & 0.506602 & 0.590194 & 0.663214 & 0.728099  \\
        \hline
    \end{tabular}
    \caption{Search cost for $\msp_n(\bone)$ and corresponding optimizer expansion factors $k=\kappa_n$ values for trajectories $\{f_i^k\}_{i\in[n]}$, when $n=1,\ldots,7$.}
    \label{tab:transposed}
\end{table}

Table~\ref{tab:transposed} provides some indicative numerical values for search costs $\mathcal U_n$, for small values of $n$. 
The values $\mathcal U_1$ and $\kappa_1$ are the known values that have given the provably optimal algorithm to $\mathcal P\text{-}\msp_1(1)$ due to~\cite{langetepe2010optimality}.
Below, we explore the behavior of $\mathcal U_n,\kappa_n$ as a function of $n$, which will be useful later. 
\begin{proposition}
\label{prop: monotonicity of Cn}
Expression $\mathcal U_n$ of Theorem~\ref{thm: spiral with n unit speed agents} is strictly decreasing in $n$, and $\lim_{n\rightarrow \infty} \mathcal U_n = 1$. 
Moreover, the optimizer expansion factor $\kappa_n$ is positive, it is strictly increasing, and strictly concave in $n$, with 
$\kappa_n = \Theta \left(n^{1/3}\right)$.
\end{proposition}

\begin{proof}
Extend $\kappa_n$ to a function $\kappa(x)$ on $x>0$ by
$$
\kappa(x)^3+\kappa(x)=\frac{x}{2\pi}.
$$
Implicit differentiation gives
$$
\kappa'(x)=\frac{1}{2\pi(3\kappa(x)^2+1)}>0
$$
and
$$
\kappa''(x)=-\frac{6\kappa(x)(\kappa'(x))^2}{3\kappa(x)^2+1}<0.
$$
Thus, $\kappa(x)$ is positive, strictly increasing, and strictly concave. The defining equation also gives $\kappa_n=\Theta(n^{1/3})$.

At the optimizer, $n=2\pi(\kappa_n^3+\kappa_n)$, and hence
$$
\frac{2\pi\kappa_n}{n}=\frac{1}{1+\kappa_n^2}.
$$
Therefore,
$$
\mathcal U_n
=
\frac{\sqrt{1+\kappa_n^2}}{\kappa_n}
\exp\left(\frac{1}{1+\kappa_n^2}\right).
$$
For $x>0$, set
$$
V(x)=\frac{\sqrt{1+x^2}}{x}\exp\left(\frac{1}{1+x^2}\right).
$$
A direct calculation yields
$$
\frac{\dd}{\dd x}\ln V(x)
=-\frac{3x^2+1}{x(1+x^2)^2}<0.
$$
Since $\kappa_n$ is strictly increasing, $\mathcal U_n=V(\kappa_n)$ is strictly decreasing. Finally, $\kappa_n\to\infty$ and $\kappa_n/n\to0$, so
$$
\lim_{n\to\infty}\mathcal U_n=1.
$$
\qed \end{proof}

\section{Searching the Plane with $n$ Agents of Arbitrary Speeds}
\label{sec: arbitrary speeds spiral}

In this section, we prove Theorem~\ref{thm: simplified general spiral upper bound}, providing upper bounds for the problem $\mathcal P\text{-}\msp_n(c)$, where $n$ agents of arbitrary positive speeds search the plane.  
First, in Section~\ref{sec: off-set algo}, we introduce the \emph{Off-Set Spiral Trajectory}, a general algorithm for the problem, presented as Algorithm~\ref{algo: Off-Set Trajectory}.  
Then, in Section~\ref{sec: off-set general properties}, we analyze its general properties with respect to its parameters and the input speeds $c \in \mathbb{R}_{>0}^n$.  
Building on these observations, Section~\ref{sec: off-set optimal parameters} determines an appropriate subset of agents and selects optimal parameters for running the Off-Set algorithm. This leads to Theorem~\ref{thm: opt off-set algo and performance} (page~\pageref{thm: opt off-set algo and performance}), which establishes certain conditions on the input speeds and upper bounds for $\mathcal P\text{-}\msp_n(c)$. The parameters used in the theorem are computed by Algorithm~\ref{algo: Off-Set Trajectory Parameter Selection} and serve as input to Algorithm~\ref{algo: Off-Set Trajectory}.  
Finally, in Section~\ref{sec: simplified general spiral upper bound}, we simplify the premises of Theorem~\ref{thm: opt off-set algo and performance}, leading to the proof of Theorem~\ref{thm: simplified general spiral upper bound}.  

Our construction is a natural extension of spiral-based search to heterogeneous speeds, under the strong balancing requirement that all participating agents incur the same competitive cost. One of the goals of the paper is to understand whether this natural generalization remains optimal once agent speeds differ.

\subsection{The Speed-Dependent Off-Set Spiral Algorithm for $\mathcal P\text{-}\msp_n(c)$}
\label{sec: off-set algo}

In this section, we generalize the trajectory $\{f_i^k\}_{i\in[n]}$ for $\msp_n(\bone)$ in order to propose an explicit parameter-selection algorithm giving an upper bound for $\msp_n(c)$ for arbitrary $c\in \reals_{>0}^n$. 
Recall that the speeds $\{c_i\}_{i\in[n]}$ are given in non-increasing order and that $c_0=1$.  As a reminder, we also refer to agent-$0$ as agent-$n$, and we use $c_0=c_n=1$. 
In that direction, we propose a family of spiral-type search algorithms, as described in Algorithm~\ref{algo: Off-Set Trajectory}, which we call \emph{Off-Set Trajectories}. Under the strict parameter inequalities identified below, the participating trajectories are pairwise non-intersecting. 
The key difference from Algorithm~\ref{algo: uniform spiral trajectory} is that agents in the new trajectories are assigned arbitrary angular separations, parameterized by values $\phi_i$, which will later be defined in terms of the agents' speeds. 
These degrees of freedom, in addition to the spiral expansion parameter $k$, allow us to balance first-discovery costs across participating agents. The product argument in Section~\ref{sec: simplified general spiral upper bound} proves optimality within this common-expansion-factor family.

\begin{algorithm}[H]
\caption{Off-Set Trajectory for $\mathcal P\text{-}\msp_n(c)$}
\label{algo: Off-Set Trajectory}
\begin{algorithmic}
\REQUIRE 
$(\phi_0,\ldots,\phi_n)\in\reals^{n+1}$ and $k\in \reals$, satisfying $2\pi = \phi_0 \geq \phi_1 \geq \ldots \geq \phi_{n-1} \geq \phi_n = 0$ and $k>0$.
\STATE \textbf{Output Trajectories:}  
$
g_i^k(t) = 
\langle 
c_i \cdot e^{kt},
t+\phi_i
\rangle,
\hspace{0.5em} t \in \reals, \hspace{0.5em} i \in [n].
$
\end{algorithmic}
\end{algorithm}

Here, the common-expansion-factor Off-Set family consists of full-speed trajectories of the displayed form with a single chosen $k>0$, arbitrary angular phases modulo $2\pi$, and any nonempty subset of the agents. The ordered phase convention in Algorithm~\ref{algo: Off-Set Trajectory} describes the construction we use; the lower bound later also allows other phase assignments.

For any fixed parameter value $t$, the distance of agent $i$ from the origin equals $c_i e^{kt}$, and thus the distances of any two agents remain proportional. Consequently, as $t$ varies, the point configuration $\{g_i^k(t)\}_{i\in[n]}$ evolves by uniform scaling and rotation, and any two such configurations are similar. This contrasts with the unit speed case, where all agents have the same distance from the origin at each snapshot, and the configuration forms a regular $n$-gon. Thus, in the multispeed setting, the configuration need not be regular, but its shape is preserved over time up to similarity, see Figure~\ref{fig: nonuniform spiral}.

\begin{figure}[h!]
    \centering
    \includegraphics[width=0.42\textwidth]{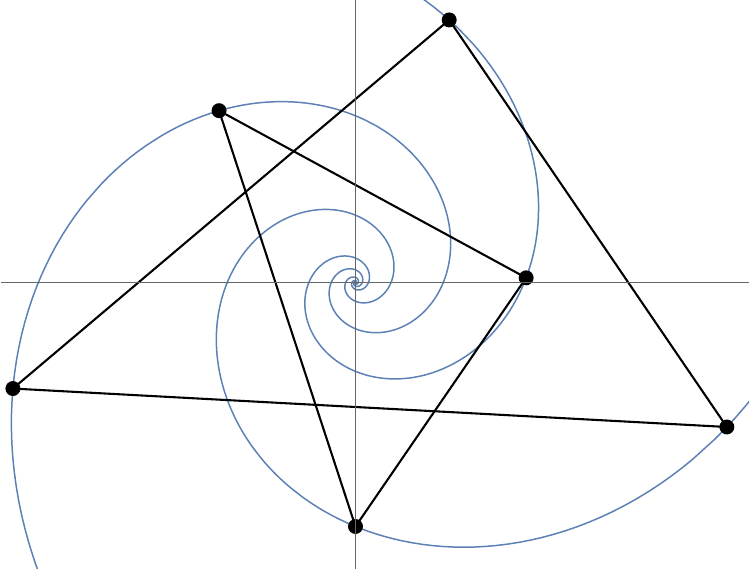}
    \caption{
    The depiction of trajectories of Algorithm~\ref{algo: Off-Set Trajectory}. 
    In this example, we use $n=3$ agents with speeds $c_0=1$, $c_1=0.9$, and $c_2=0.7$, and we plot the corresponding offset logarithmic spirals (as defined by the algorithm). 
    We also show two snapshots of the agents, at two distinct parameter values, together with the triangles formed by connecting the three agent positions in each snapshot. 
    In contrast to the unit speed case (Figure~\ref{fig: uniform spiral}), where all agents lie on the same circle at every snapshot and hence form a regular $n$-gon, 
    the multispeed agents lie at different distances from the origin. These distances remain proportional, and the two triangles are similar (their side lengths scale by the same factor between snapshots).
    }
    \label{fig: nonuniform spiral}
\end{figure}

Note that the parameters $\phi_0 = 2\pi$ and $\phi_n = 0$ are fixed and correspond to agent-0 (also referred to as agent-$n$ for notational convenience). If the same formula is extended to the bookkeeping index $n$ using $c_n=c_0$ and $\phi_n=0$, then periodicity gives $g_n^k(t)=g_0^k(t)$, ensuring cyclic behavior in the notation.
As before, the condition $k > 0$ is necessary for the correctness of the search algorithm, guaranteeing that every point in the plane is eventually exposed. The ordering  
$
\phi_0 \geq \phi_1 \geq \dots \geq \phi_{n-1} \geq 0
$  
is a normalization that determines the appropriate sequence in which agents should be considered based on their speeds.  
With strictly decreasing phases, each ray is visited in cyclic order by agents $0,1,2,\dots,n-1$, repeating indefinitely. 
Equal phases modulo $2\pi$ correspond to simultaneous visits (within each such group, keeping only a farthest agent preserves every discovery time).

The main contribution of this section is a parameter-selection algorithm that takes as input $\{c_i\}_{i\in [n]}$ and determines parameters for the Off-Set Algorithm~\ref{algo: Off-Set Trajectory}, yielding strong upper bounds for $\mathcal P\text{-}\msp_n(c)$. Notably, we show that for an arbitrary input $c \in \reals_{>0}^n$, not all agents necessarily expose new points in a solution. Thus, minimizing the search cost of the Off-Set algorithm for $\msp_n(c)$ reduces to identifying a subset $A \subseteq [n]$ and parameters $k > 0$ and $\{\phi_i\}_{i \in A}$ that minimize the search cost.

To compute these parameters, we propose Algorithm~\ref{algo: Off-Set Trajectory Parameter Selection}, which utilizes precomputed expansion factors $\kappa_j$ from Lemma~\ref{lem: closed form of k_i} and the search cost $\mathcal U_j$ for $j$ unit-speed agents from Theorem~\ref{thm: spiral with n unit speed agents}. Algorithm~\ref{algo: Off-Set Trajectory Parameter Selection} returns parameters $\phi_j$ only for the agents that actively contribute to the search. Specifically, it identifies a key agent and all agents at least as fast, while leaving slower agents unused. The values of $\phi_j$ for these ignored agents are irrelevant to the solution.

\begin{algorithm}[H]
\caption{Off-Set Trajectory Parameter Selection}
\label{algo: Off-Set Trajectory Parameter Selection}
\begin{algorithmic}[1]
\REQUIRE $c\in \reals_{>0}^n$ of $\msp_n(c)$, with $c_0=1$ and $c_i$ non-increasing in $i\in [n]$.
\STATE \label{step:feasible_set} Set
$$
\mathcal I:=\left\{m\in[n]:c_m\geq e^{-2\pi\kappa_{m+1}/(m+1)}\mathcal G_{m+1}\right\}.
$$
\STATE \label{step:set_mu} Set
$$
\text{choose }\mu\in\argmin_{m\in\mathcal I}\frac{\mathcal U_{m+1}}{\mathcal G_{m+1}}
$$
and $\nu\leftarrow\mu+1$.
\STATE \label{step:output} \textbf{Output:} $\mu$, $\nu$, and
$$
\phi_j=
\frac{2\pi(\nu-j)}{\nu}-\frac{j}{\kappa_\nu}\ln\left(\frac{\mathcal G_j}{\mathcal G_\nu}\right),
\hspace{2em} j=0,1,\ldots,\nu.
$$
\end{algorithmic}
\end{algorithm}

The feasible set $\mathcal I$ is nonempty, since $0\in\mathcal I$. Using cumulative products, the algorithm examines at most $n$ prefixes once the speeds are ordered.
The goal in subsequent Sections~\ref{sec: off-set general properties} and~\ref{sec: off-set optimal parameters} is to prove the next theorem, which can be seen as precursor to Theorem~\ref{thm: simplified general spiral upper bound}. 

\begin{theorem}
\label{thm: opt off-set algo and performance}
Let $\mu$ and $\nu=\mu+1$ be the indices computed by Algorithm~\ref{algo: Off-Set Trajectory Parameter Selection} on input $c \in \mathbb{R}_{>0}^n$, along with parameters $\{\phi_i\}_{i=0}^{\nu}$.  
Then, we have $\mu \geq 0$, which implies $\nu \geq 1$, and the parameters $\{\phi_i\}_{i=0}^{\nu}$ along with expansion factor $k=\kappa_{\nu}$ are valid for the Off-Set Algorithm~\ref{algo: Off-Set Trajectory} applied to the first $\nu$ agents of $\mathcal P\textrm{-}\msp_n(c)$; the remaining agents are unused.  
Moreover, the induced search cost is
$\mathcal U_{\nu}
/\mathcal G_{\nu}$,
where the value $\mathcal U_\nu$ is the search cost for searching with $\nu$ unit-speed agents, and is given in Theorem~\ref{thm: spiral with n unit speed agents}.
\end{theorem}

\subsection{Properties of the Off-Set Trajectories}
\label{sec: off-set general properties}

In this section, we identify and analyze key properties of Algorithm~\ref{algo: Off-Set Trajectory}. As throughout the paper, we assume that the agent speeds $c_i$ are given in non-increasing order.
We begin by showing that agents following the trajectories $\{g_i^k(t)\}_{i \in [n]}$ search the plane in an invariant pattern.

\begin{lemma}
\label{lem: preserved lengths}
For the Off-Set search strategy $\{g^k_i\}_{i\in[n]}$ to $\msp_n(c)$, and for every $i \in [n]$, we have $\frac{1}{c_i} \ell_i(t_0) = \ell_0(t_0)$, for all $t_0 \in \reals$. 
As a result, at any moment of the execution of the search strategy, any two agents maintain the same angular distance. 
\end{lemma}

\begin{proof}
Fix some $t_0 \in \reals$, and agent $i$ with speed $c_i$. 
We observe that the angular distance between points $g_0^k(t_0)$ and $g_i^k(t_0)$ equals $2\pi-\phi_i$. Hence, the two agents preserve their angular distance if and only if they reach the corresponding points $g_0^k(t_0)$ and $g_i^k(t_0)$ at the same time. 
We argue next that the length of trajectory $g^k_i$ up to $t=t_0$ is $\ell_i(t_0) = c_i \ell_0(t_0)$. 
Indeed, let $h:\reals\to \reals^2$ be as in the statement of Lemma~\ref{lem: standard integral}. We note that $h(t)=g_0^k(t)$, since $\phi_0=2\pi$ and $c_0=1$. 
But then, by the same lemma, and for any speed-$c_i$, we have 
$$
\ell_i(t_0) = 
\int_{-\infty}^{t_0}
\norm{
\dd g_i^k(t) / \dd t
}
\dd t
= 
c_i 
\int_{-\infty}^{t_0}
\norm{
\dd h(t) / \dd t
}
\dd t
= 
c_i 
\int_{-\infty}^{t_0}
\norm{
\dd g_0^k(t) / \dd t
}
\dd t
= 
c_i 
\ell_0(t_0),
$$
as promised. 
But then, the speed-$c_i$ agent reaches point $g_i^k(t_0)$ in time $\frac1{c_i} \ell_i(t_0)= \frac{1}{c_i}c_i \ell_0(t_0)= \ell_0(t_0)$, which is the same as the time that the unit-speed agent needs to reach point $g_0^k(t_0)$. 
That is, the angular distance between agent-$i$ and agent-$0$ is preserved, hence the angular distance between any two agents is preserved. 
\qed \end{proof}

In the next two lemmas, we assume that the angular offsets $\phi_i$ are all distinct. If two agents reach the same ray at the same time, only the one reaching farther from the origin can expose new points (if they reach the same point, one of them can be ignored). We say that an agent is active if it is the first agent to expose a nonempty interval of points on some ray. The case in which some angular offsets coincide in the equalizing construction will be treated separately below.

\begin{lemma}
\label{lem: off-set spiral for n agents}
For the Off-Set search strategy $\{g^k_i\}_{i\in[n]}$ to $\msp_n(c)$ 
with $2\pi=\phi_0>\phi_1>\cdots>\phi_n=0$ and $k>0$, all agents are active if and only if
$$
c_{i} e^{- k \phi_{i}} < c_{i+1} e^{- k \phi_{i+1}}, i\in [n].
$$
Here, agent $n$ is identified as agent $0$ and $c_0=1$, $\phi_n=0$.
Moreover, if the above condition holds, then the $n$ trajectories are pairwise disjoint, and for every $s\in \integers$ and $i\in[n]$, agent $i+1$ is the first to expose all points $(x,0)\in \reals^2$ on the ray $\mathcal L_0$, where 
$$
x\in \left(
 c_i e^{2\pi k s - k \phi_i}, 
 c_{i+1} e^{2\pi k s  - k \phi_{i+1}}
\right]. 
$$
\end{lemma}

\begin{proof}
For $i=0,\ldots,n$, set
$
a_i:=c_i e^{-k\phi_i}
$
and fix $s\in\mathbb Z$. Agent $i$ reaches $\mathcal L_0$ when the trajectory parameter equals $2\pi s-\phi_i$, and at that point its distance from the origin is

$$
c_i e^{k(2\pi s-\phi_i)}=e^{2\pi ks}a_i.
$$

Since the values $\phi_i$ are all distinct and given in decreasing order, the agents reach $\mathcal L_0$ successively in the order $0,1,\ldots,n$.

Suppose now that $a_{i+1}\leq a_i$. Then, when agent $i+1$ reaches $\mathcal L_0$, it does not reach farther than agent $i$ did in the preceding visit. Hence agent $i+1$ does not expose any new point on $\mathcal L_0$. 
The same argument holds for every $s\in\mathbb Z$, since increasing $s$ by one multiplies all these radii by the same factor $e^{2\pi k}$, and after rotating the trajectories, to every ray. Therefore, if every agent exposes new points, we must have
$
a_0<a_1<\cdots<a_n.
$

Conversely, assume that
$
a_0<a_1<\cdots<a_n.
$
Then the successive visits to $\mathcal L_0$ reach strictly farther from the origin. Moreover,
$
a_n=e^{2\pi k}a_0,
$
so the same ordering continues from one revolution to the next. It follows that, immediately before agent $i+1$ reaches $\mathcal L_0$, the farthest point already exposed on this ray is the point reached by agent $i$. Hence agent $i+1$ is the first agent to expose exactly the interval stated in the lemma.

To prove the disjointness claim, set $a_i=c_i e^{-k\phi_i}$ for $i=0,\ldots,n$. Under the strict inequalities,
$$
e^{-2\pi k}=a_0<a_1<\cdots<a_n=1.
$$
If two trajectories $i<j<n$ intersected, comparison of their angular coordinates and radii would give
$$
\frac{a_i}{a_j}=e^{-2\pi km}
$$
for some $m\in\integers$. However, $e^{-2\pi k}<a_i/a_j<1$, whereas no integer $m$ makes $e^{-2\pi km}$ lie strictly in this interval. Thus, the trajectories are pairwise disjoint.
\qed
\end{proof}

We now calculate the search cost of the Off-Set algorithm under specific conditions on its parameters and input speeds. Define the \emph{agent-$i$ search cost} by
$$
S_i:=\sup\left\{\frac{T_i(A)}{\norm A}:A\in\mathcal P,\ T_i(A)=\min_j T_j(A)<\infty\right\},
$$
with value $0$ if the set is empty. The overall cost is $\max_i S_i$.
In the next lemma, every agent is active and phases are strict, so each target has a unique first-exposing agent.

\begin{lemma}
\label{lem: search cost with assumptions}
Consider the Off-Set search strategy $\{g^k_i\}_{i\in[n]}$ with $k>0$, strictly ordered offsets $2\pi=\phi_0>\cdots>\phi_n=0$, and $c_i e^{-k\phi_i}<c_{i+1}e^{-k\phi_{i+1}}$ for every $i\in[n]$. Then the agent-$(i+1)\bmod n$ first-discovery search cost equals
$$
\frac{\sqrt{1+k^2}}{k} \frac{1}{c_{i} }
e^{-k (\phi_{i+1}-\phi_{i})},~~
i\in [n],
$$
where agent-$n$ is identified as agent-$0$, $c_n=c_0=1$, $\phi_0=2\pi$, and $\phi_n=0$.
Moreover, if the agent-$i$ search cost is the same for all agents $i\in [n]$, then the search cost of the search strategy equals 
$$
\frac{\sqrt{1+k^2}}{k}
\frac
{
e^{2\pi k/n}
}
{
\mathcal G_n
}.
$$
\end{lemma}

\begin{proof}
Let $R_\eta$ denote rotation by angle $\eta$. We have
$
e^{-k\psi}R_{-\psi}g_i^k(t)=g_i^k(t-\psi).
$
Hence, after rotating the plane by $-\psi$ and scaling distances by $e^{-k\psi}$, a target on $\mathcal L_\psi$ is mapped to $\mathcal L_0$. Since both the distance travelled by the agent and the distance of the target from the origin are multiplied by the same factor, the search cost does not change. Therefore, without loss of generality, we may assume that the target is at $(x,0)\in\mathcal L_0$, where $x>0$.
 
Since $k > 0$, there exists $s \in \integers$ such that  
$$
x \in \left(e^{2\pi k s - 2\pi k}, e^{2\pi k s} \right].
$$  
Now, fix $i \in [n]$. By Lemma~\ref{lem: off-set spiral for n agents}, agent-$(i+1)$ exposes all points  
$
x \in \left(
c_{i} e^{2\pi k s  - k \phi_{i}} , c_{i+1} e^{2\pi k s - k \phi_{i+1}}
\right],
$  
where $c_0 = 1$, $\phi_0 = 2\pi$, and agent-$n$ is identified with agent-$0$. For notational convenience, denote this interval as $\mathcal{I}^s_i$.  
By Lemma~\ref{lem: preserved lengths}, agent-$(i+1)$ exposes $(x,0)$ at time  
$$
\frac{\ell_{i+1}(2\pi s - \phi_{i+1})}{c_{i+1}} = \ell_0(2\pi s - \phi_{i+1}).
$$  
Thus, the search cost for agent-$(i+1)$ is given by  
$$
\sup_{s\in \integers}
\sup_{ x\in \mathcal{I}^s_i } 
\frac{\ell_0(2\pi s - \phi_{i+1})}{x}
\stackrel
{
\text{Lemma~\ref{lem: standard integral}}
}
{=}
\sup_{s\in \integers}
\frac{\sqrt{1+k^2}}{k} 
\frac{e^{2\pi k s - k \phi_{i+1}}}
{
c_{i} e^{2\pi k s  - k \phi_{i}}
}
=
\frac{\sqrt{1+k^2}}{k} \frac{1}{c_{i} }\frac{e^{-k \phi_{i+1}}}{e^{-k \phi_{i}}},
$$  
concluding the first claim of the lemma.  

Now, assume that the agent-$i$ search cost, denoted by $S_i$, is invariant for all $i \in [n]$, meaning $S_i = S$ for some search cost $S$.  
Since agent-$0$ is identified with agent-$n$, and we have $c_0 = c_n = 1$ and $\phi_0 = 2\pi$, $\phi_n = 0$, it follows that  
$$
S^n 
= \prod_{i=1}^{n} S_i 
= 
\left( \frac{\sqrt{1+k^2}}{k} \right)^n
 \frac{1}{\prod_{i=0}^{n-1} c_{i} }
\prod_{i=0}^{n-1}
 \frac{e^{-k \phi_{i+1}}}{e^{-k \phi_{i}}}
=
 \left( \frac{\sqrt{1+k^2}}{k} \right)^n
 \frac{1}{\prod_{i=0}^{n-1} c_{i} }
e^{2\pi k},
$$  
which establishes the desired formula for $S$.
\qed \end{proof}

The next lemma gives a feasible equalized Off-Set trajectory. When its parameters are equal, redundant agents can be removed without changing the overall cost (the remaining agents have equal first-discovery costs).
\begin{lemma}
\label{lem: def of phi that make all agents expose points and cost ind of agent}
Consider $n$ agents satisfying
$$
\min_{i\in[n]}c_i\geq e^{-2\pi k/n}\mathcal G_n
$$
for some $k>0$. Set $c_n=c_0=1$ and, for $j=0,1,\ldots,n$, define
\begin{equation}
\label{equa: equializers phi}
\phi_j=\frac{2\pi(n-j)}{n}-\frac{j}{k}\ln\left(\frac{\mathcal G_j}{\mathcal G_n}\right).
\end{equation}
Then $2\pi=\phi_0\geq\phi_1\geq\cdots\geq\phi_n=0$, and these parameters define a feasible Off-Set trajectory. Its search cost is
$$
\frac{\sqrt{1+k^2}}{k}\frac{e^{2\pi k/n}}{\mathcal G_n}.
$$
If all speed inequalities are strict, all phases are strict and every agent is active. Otherwise, set $b=e^{-2\pi k/n}\mathcal G_n$. 
Then, deleting all agents with speed $b$ preserves discovery times, and every remaining agent has the displayed first-discovery cost.
\end{lemma}

\begin{proof}
Set
$
b:=e^{-2\pi k/n}\mathcal G_n.
$
From~\eqref{equa: equializers phi}, we have $\phi_0=2\pi$, $\phi_n=0$, and
$$
\phi_j-\phi_{j+1}
=
\frac{1}{k}\ln\left(\frac{c_j}{b}\right)
\geq 0.
$$

Hence, $
2\pi=\phi_0\geq\phi_1\geq\cdots\geq\phi_n=0,
$
so the angular parameters are valid. Moreover, by Lemma~\ref{lem: standard integral}, the length of every trajectory from the origin to any point $g_i^k(t)$ is finite. Since $k>0$, every point in the plane is eventually exposed, and hence the resulting trajectories give a feasible solution.

We now calculate the search cost. Since the speeds are given in non-increasing order and $c_i\geq b$ for every $i\in[n]$, there is an index $r\in[n]$ such that
$
c_0\geq c_1\geq\cdots\geq c_r>b,
$
while, if $r<n-1$,
$$
c_{r+1}=\cdots=c_{n-1}=b.
$$
Such an index always exists, since
$$
b=e^{-2\pi k/n}\mathcal G_n<1=c_0.
$$
If $r<n-1$, then 
$
\phi_{r+1}=\phi_{r+2}=\cdots=\phi_n=0.
$
Therefore, every agent $j=r+1,\ldots,n-1$ is, at every time, on the same ray as agent $0$, while its distance from the origin is only a factor $b<1$ of the distance of agent $0$. Such an agent never exposes a point before agent $0$, and hence all agents with speed $b$ can be ignored without changing any target discovery time.

It remains to calculate the search cost induced by agents $0,\ldots,r$. Consider first $j=0,\ldots,r-1$. When agent $j+1$ reaches the same ray after agent $j$, the difference between the corresponding trajectory parameters is
$$
\phi_j-\phi_{j+1}
=
\frac{1}{k}\ln\left(\frac{c_j}{b}\right).
$$
The ratio between the distances from the origin reached in these two consecutive visits is therefore
$$
\frac{c_{j+1}}{c_j}
e^{k(\phi_j-\phi_{j+1})}
=
\frac{c_{j+1}}{b}
>1.
$$
Thus, agent $j+1$ exposes new points. Applying the same calculation as in Lemma~\ref{lem: search cost with assumptions}, its search cost is
$$
S_{j+1}
=
\frac{\sqrt{1+k^2}}{k}
\frac{e^{k(\phi_j-\phi_{j+1})}}{c_j}
=
\frac{\sqrt{1+k^2}}{kb}.
$$
After agent $r$ reaches a ray, the next agent among $0,\ldots,r$ to reach the same ray is agent $0$. Since all angular differences after $\phi_r$ are zero, we also have
$$
\phi_r
=
\frac{1}{k}\ln\left(\frac{c_r}{b}\right).
$$
The same calculation gives
$$
S_0
=
\frac{\sqrt{1+k^2}}{k}
\frac{e^{k\phi_r}}{c_r}
=
\frac{\sqrt{1+k^2}}{kb}.
$$

This calculation also covers the case $r=0$, in which agent $0$ is the only agent that cannot be ignored and $\phi_0=2\pi$.
Therefore, every agent that exposes new points induces the same search cost, and
$$
\frac{\sqrt{1+k^2}}{kb}
=
\frac{\sqrt{1+k^2}}{k}
\frac{e^{2\pi k/n}}{\mathcal G_n},
$$
as claimed.
\qed
\end{proof}


\subsection{Proof \& Discussion of Theorem~\ref{thm: opt off-set algo and performance}}
\label{sec: off-set optimal parameters}
This section is devoted to proving Theorem~\ref{thm: opt off-set algo and performance}, which establishes the correctness of the Off-Set Algorithm~\ref{algo: Off-Set Trajectory} when using the parameters selected by Algorithm~\ref{algo: Off-Set Trajectory Parameter Selection} for input speeds $c \in \reals_{>0}^n$.  
We first prove that Algorithm~\ref{algo: Off-Set Trajectory Parameter Selection} always returns a valid index $\mu \in [n]$. Then, we verify that the returned parameters are valid inputs to the Off-Set Algorithm~\ref{algo: Off-Set Trajectory} and ensure a feasible search trajectory. Finally, we compute the search cost for $\msp_n(c)$, completing the proof of Theorem~\ref{thm: opt off-set algo and performance}.

\begin{proof}[of Theorem~\ref{thm: opt off-set algo and performance}]
The set $\mathcal I$ in Step~\ref{step:feasible_set} is nonempty because
$$
c_0=\mathcal G_1=1\geq e^{-2\pi\kappa_1}.
$$
Hence, Algorithm~\ref{algo: Off-Set Trajectory Parameter Selection} returns some $\mu\in[n]$ and $\nu=\mu+1$.

Since $\mu\in\mathcal I$ and the speeds are non-increasing,
$$
\min_{i\in[\nu]}c_i=c_\mu\geq e^{-2\pi\kappa_\nu/\nu}\mathcal G_\nu.
$$
Lemma~\ref{lem: def of phi that make all agents expose points and cost ind of agent}, applied to the first $\nu$ agents with $k=\kappa_\nu$, shows that the returned angles are valid and that the induced search cost is
$$
\frac{\sqrt{1+\kappa_\nu^2}}{\kappa_\nu}
\frac{e^{2\pi\kappa_\nu/\nu}}{\mathcal G_\nu}
=
\frac{\mathcal U_\nu}{\mathcal G_\nu}.
$$
The remaining agents may stay at the origin.
\qed \end{proof}

To better understand the calculations performed by Algorithm~\ref{algo: Off-Set Trajectory Parameter Selection}, we provide a high-level overview of its operation on input $c \in \reals_{>0}^n$.  
The algorithm effectively selects an index $\mu \in [n]$, meaning that only agents at least as fast as agent-$\mu$ will participate in exposing new points.  
If agents are sorted in non-increasing order of speed, the participating agents correspond to $\nu = \mu+1$, that is, agents $j$ with $j = 0, \ldots, \mu$.  
The returned angular parameters $\phi_j$ are chosen so that the agent search cost is independent of the agent index, ensuring that the points exposed by each agent induce the same cost.  

A crucial aspect of Algorithm~\ref{algo: Off-Set Trajectory Parameter Selection} is that it filters out agents with insufficiently fast speeds.  
Among the feasible prefixes, it selects a fastest-agent prefix, denoted by $A$, to minimize the ratio  
$
\mathcal{U}_{|A|} / \mathcal{G}_{|A|}.
$  

An interesting observation is that Theorem~\ref{thm: opt off-set algo and performance} generalizes Theorem~\ref{thm: spiral with n unit speed agents}, whose proof also serves as a key building block in establishing the former result.
In particular, consider $n$ unit-speed agents, meaning $c_i = 1$ for all $i \in [n]$ is given as input to Algorithm~\ref{algo: Off-Set Trajectory Parameter Selection}.  
These agents are already sorted in non-increasing order, and their geometric means satisfy $\mathcal{G}_m = 1$ for every $m=0,\ldots,n$. Hence, every $m\in[n]$ belongs to the feasible set in Step~\ref{step:feasible_set}, because
$$
c_m=1\geq e^{-2\pi\kappa_{m+1}/(m+1)}\mathcal G_{m+1}.
$$
By Proposition~\ref{prop: monotonicity of Cn}, the values $\mathcal{U}_i$ are strictly decreasing in $i$, so the smallest search cost, as determined by Step~\ref{step:set_mu}, is achieved at $\mu=n-1$. Then, in Step~\ref{step:output}, the algorithm sets $\nu=\mu+1=n$, and the angular parameters are
$$
\phi_j=\frac{2\pi(n-j)}{n},
\hspace{2em} j=0,\ldots,n,
$$
which exactly matches the definition in Algorithm~\ref{algo: uniform spiral trajectory}.  

Thus, Theorem~\ref{thm: opt off-set algo and performance} not only confirms the validity of the Off-Set algorithm, but also generalizes Theorem~\ref{thm: spiral with n unit speed agents}. In particular, the technical proof of Theorem~\ref{thm: opt off-set algo and performance}, which provides upper bound guarantees for arbitrary speed agents, builds directly on the analysis underlying Theorem~\ref{thm: spiral with n unit speed agents} to re-establish the performance guarantee of $\mathcal{U}_n$ for the special case of uniform unit-speed agents.

\subsection{A Simplified Algorithm for $\mathcal P\text{-}\msp_n(c)$}
\label{sec: simplified general spiral upper bound}

This section proves Theorem~\ref{thm: simplified general spiral upper bound}. The next inequality is the only technical point needed to show that a prefix attaining the minimum in that theorem is feasible for the Off-Set construction.

\begin{lemma}
\label{lem: critical ratio at least 1}
Let $\kappa_n$ be as in Lemma~\ref{lem: closed form of k_i}, and set
$$
A(x):=\frac{\sqrt{1+x^2}}{x},
\hspace{2em} x>0.
$$
Then, for every integer $l\geq1$,
$$
\left(\frac{A(\kappa_{l+1})}{A(\kappa_l)}\right)^l
 e^{2\pi(\kappa_{l+1}-\kappa_l)}>1.
$$
\end{lemma}

\begin{proof}
Set $x=\kappa_l$ and $y=\kappa_{l+1}$. Since
$$
\frac{\dd}{\dd t}\ln A(t)=-\frac{1}{t(1+t^2)},
$$
we have
\begin{align*}
\ln\left[
\left(\frac{A(y)}{A(x)}\right)^l e^{2\pi(y-x)}
\right]
&=
l\left(\ln A(y)-\ln A(x)\right)+2\pi(y-x)\\
&=
2\pi\left[
(y-x)-x(1+x^2)\int_x^y\frac{\dd t}{t(1+t^2)}
\right],
\end{align*}
where we used $l=2\pi(x^3+x)$. Since $t(1+t^2)>x(1+x^2)$ for $t\in(x,y]$, the quantity in brackets is positive.
\qed \end{proof}

\begin{proof}[of Theorem~\ref{thm: simplified general spiral upper bound}]
Choose
$$
\nu\in\argmin_{i=1,\ldots,n}\frac{\mathcal U_i}{\mathcal G_i}
$$
and set $\ell=\nu-1$. If $\ell=0$, then $0\in\mathcal I$ by the proof of Theorem~\ref{thm: opt off-set algo and performance}. Assume $\ell\geq1$. By the minimality of $\nu$,
$$
\frac{\mathcal U_{\ell+1}}{\mathcal G_{\ell+1}}
\leq
\frac{\mathcal U_\ell}{\mathcal G_\ell}.
$$
Using $\mathcal G_{\ell+1}^{\ell+1}=\mathcal G_\ell^\ell c_\ell$, this is equivalent to
$$
c_\ell\geq
\left(\frac{\mathcal U_{\ell+1}}{\mathcal U_\ell}\right)^{\ell+1}\mathcal G_\ell.
$$
Lemma~\ref{lem: critical ratio at least 1} is equivalent to
$$
\left(\frac{\mathcal U_{\ell+1}}{\mathcal U_\ell}\right)^\ell
> e^{-2\pi\kappa_{\ell+1}/(\ell+1)}.
$$
Raising this inequality to the power $(\ell+1)/\ell$ gives
$$
\left(\frac{\mathcal U_{\ell+1}}{\mathcal U_\ell}\right)^{\ell+1}
> e^{-2\pi\kappa_{\ell+1}/\ell}.
$$
Consequently,
$
c_\ell> e^{-2\pi\kappa_{\ell+1}/\ell}\mathcal G_\ell,
$
which is equivalent to
$$
c_\ell> e^{-2\pi\kappa_{\ell+1}/(\ell+1)}\mathcal G_{\ell+1}.
$$
Thus the feasibility inequality is strict for every globally minimizing prefix (for $\ell=0$ it is $1>e^{-2\pi\kappa_1}$). All selected phases are therefore strict and all selected agents are active.
Thus, $\ell\in\mathcal I$. Algorithm~\ref{algo: Off-Set Trajectory Parameter Selection} therefore considers a prefix attaining the global minimum, and Theorem~\ref{thm: opt off-set algo and performance} yields the claimed search cost. The displayed condition on the slowest participating agent was already derived above.

It remains to prove optimality within the specified Off-Set family.
Fix any $k>0$, any selected agents, and any angular phase assignments. At parameter $t$, every selected agent reaches its position at physical time $A(k)e^{kt}$, where $A(k)=\sqrt{1+k^2}/k$. Along a fixed ray, visits repeat with parameter period $2\pi$ and radial factor $e^{2\pi k}$. 
Within each simultaneous group, keep only a farthest agent, and discard visits which do not increase the largest exposed radius.

Fix a ray. If an agent reaches this ray but does not reach farther from the origin than some earlier agent, then it does not expose any new point on that ray and can be ignored.
The same agents can be ignored on every ray. Indeed, changing the ray shifts all corresponding trajectory parameters by the same amount and multiplies all corresponding distances from the origin by the same positive factor. Similarly, after one complete turn, all these distances are multiplied by $e^{2\pi k}$. Therefore, whether an agent exposes new points does not depend on the particular ray or on the turn of the spiral that we consider.
Suppose that, after ignoring the agents that never expose new points, $m\geq1$ agents remain, with speeds
$
d_1,\ldots,d_m,
$
listed in the order in which they reach any fixed ray.

Let $\Delta_j>0$ be the parameter gap from agent $j$ to its successor, with indices modulo $m$.
Then $\sum_j\Delta_j=2\pi$. If the preceding visit is at parameter $t$, its radius is $d_j e^{kt}$ and the successor arrives at time $A(k)e^{k(t+\Delta_j)}$. Targets with radii decreasing to $d_j e^{kt}$ from above are first exposed by that successor, so its worst first-discovery ratio is
$$
S_{j+1}=A(k)\frac{e^{k\Delta_j}}{d_j}.
$$
Consequently,
\begin{align*}
\max_j S_j
&\geq\left(\prod_{j=1}^m S_j\right)^{1/m}
=A(k)\frac{e^{2\pi k/m}}{(\prod_{j=1}^m d_j)^{1/m}}
\geq\frac{\mathcal U_m}{(\prod_{j=1}^m d_j)^{1/m}}
\geq\frac{\mathcal U_m}{\mathcal G_m}
\geq\min_{1\leq r\leq n}\frac{\mathcal U_r}{\mathcal G_r}.
\end{align*}
Here Lemma~\ref{lem: closed form of k_i} minimizes over $k$, and the fastest $m$ agents maximize the product of any $m$ speeds. The construction already proved attains the final minimum, completing the restricted-family optimality proof (there are no implications about families of algorithms with different expansion factors, different motion directions, or unrestricted trajectories from this argument).
\qed 
\end{proof}

\section{Cone Search with a Unit-Speed Agent}
\label{sec: Cone Search}

In this section, we design a trajectory for searching for a point within an angle-$\phi$ cone, where $\phi \in (0,\pi/2)$, using a unit-speed agent. That is, we propose a solution to the problem $\mathcal{C}_\phi\text{-}\msp_1(1)$ and subsequently prove Theorem~\ref{thm: cone search upper bound}.
These results are developed not only for their independent geometric interest, but also to probe the limitations of spiral-based search when agents of different speeds are present. They will later be combined with the spiral-based constructions introduced earlier to obtain improved strategies for the multi-speed point-search problem, 
showing that the common-expansion-factor Off-Set family can be suboptimal in the multi-agent, multi-speed setting, despite the optimality of a logarithmic spiral for one agent.

\begin{figure}[h!]
    \centering
    \includegraphics[width=0.55\textwidth]{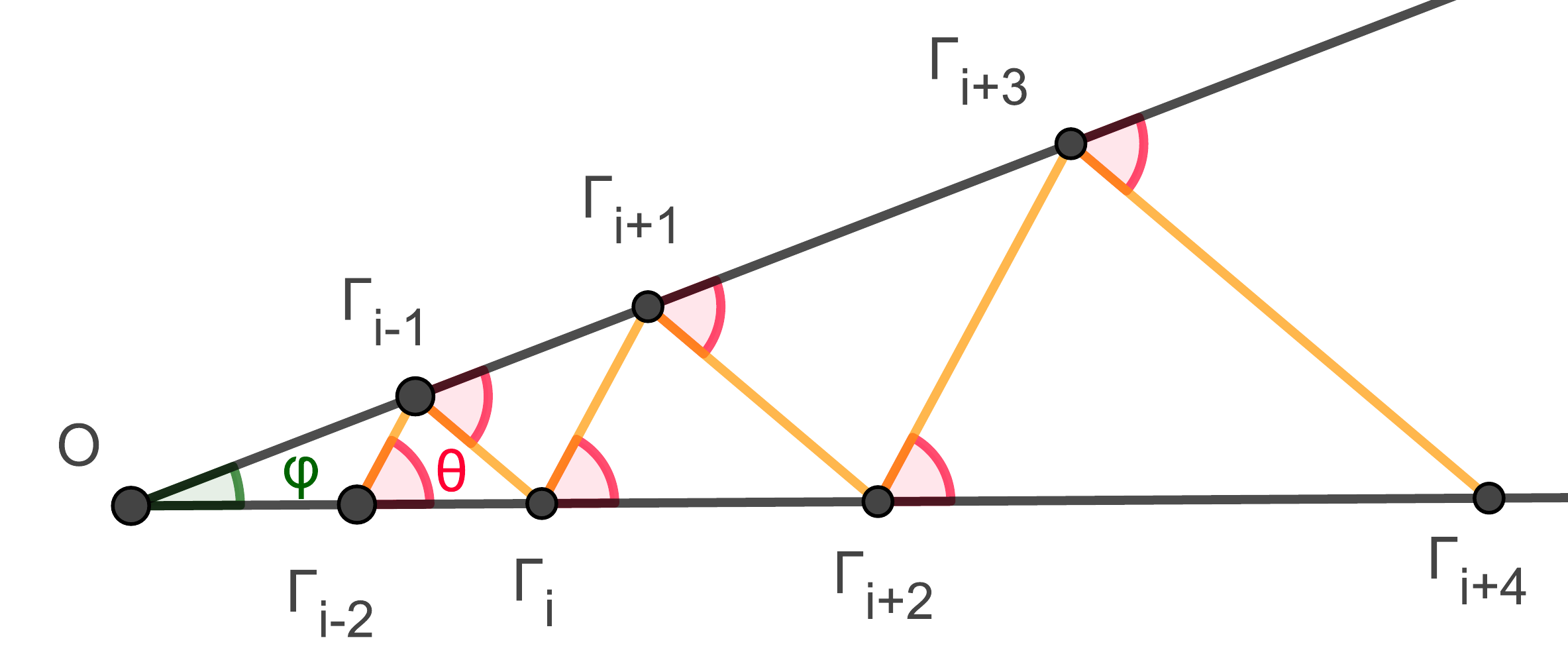}
    \caption{A depiction of the Bouncing Trajectory for $\mathcal C_\phi\text{-}\msp_1(1)$. 
    Angle $\theta$ determines the sequence of points $\Gamma_i$, as per Lemma~\ref{lem: choice of sequence g}, using sequence $q_i = \left(\tfrac{\sinn{\theta}}{\sinn{\theta-\phi}}\right)^i$. 
    }
    \label{fig:bouncing}
\end{figure}

\begin{algorithm}[H]
\caption{Bouncing Trajectory for $\mathcal C_\phi\text{-}\msp_1(1)$}
\label{algo: bouncing}
\begin{algorithmic}
\REQUIRE 
Increasing sequence $\{q_i\}_{i\in\integers}$ of positive reals, with $\lim_{i\to-\infty}q_i=0$, $\lim_{i\to\infty}q_i=\infty$, and $\sum_{i=-\infty}^{0}q_i<\infty$.
\STATE Define points 
$\Gamma_{2i} = \langle q_{2i} , 0 \rangle$, and 
$\Gamma_{2i+1} = \langle q_{2i+1} , \phi \rangle$, for $i\in \integers$. 
\STATE \textbf{Output Trajectory:}  Piece-wise linear trajectory moving between points 
$\Gamma_i \rightarrow \Gamma_{i+1}, ~i\in \integers$
\end{algorithmic}
\end{algorithm}

Algorithm~\ref{algo: bouncing} is illustrated in Figure~\ref{fig:bouncing}. The summability requirement ensures finite initial travel length. Next, we derive an upper bound on its search cost.

\begin{lemma}
\label{lem: general bouncing performance}
The Bouncing Trajectory of Algorithm~\ref{algo: bouncing} is feasible for $\mathcal C_\phi\text{-}\msp_1(1)$. Suppose additionally that, for every $i\in\integers$,
$$
\left\langle \Gamma_{i+1}-\Gamma_i,\Gamma_i\right\rangle\geq0.
$$
Then, its search cost is at most
$$
\sup_{i\in \integers} \frac{1}{\norm{\Gamma_i}} \sum_{j=-\infty}^{i+2} \norm{\Gamma_j - \Gamma_{j-1}}.
$$
\end{lemma}

\begin{proof}
Since $\norm{\Gamma_i-\Gamma_{i-1}}\leq q_i+q_{i-1}$ and $\sum_{i=-\infty}^{0}q_i<\infty$, the initial path length to every vertex is finite and tends to zero as the vertex index tends to $-\infty$. Hence the trajectory has a unit-speed execution from the origin. The limits of $q_i$ also give arbitrarily large radii on both extreme rays.
Let $\Delta_i$ be the triangle with vertices $\Gamma_i,\Gamma_{i+1},\Gamma_{i+2}$. The triangles $\Delta_i$ cover $\mathcal C_\phi$. Moreover, $\Gamma_j$ lies on the segment $O\Gamma_{j+2}$, so $\Delta_j$ is contained in the triangle $O\Gamma_{j+1}\Gamma_{j+2}$ swept out while the agent moves from $\Gamma_{j+1}$ to $\Gamma_{j+2}$. Hence, by the time the agent reaches $\Gamma_{i+2}$, every point in $\bigcup_{j\leq i}\Delta_j$ has been exposed.

Let $u_i=\Gamma_i/\norm{\Gamma_i}$. The additional assumption gives
$$
\langle \Gamma_{i+1},u_i\rangle\geq\norm{\Gamma_i}.
$$
Also, $\Gamma_{i+2}$ lies on the same extreme ray as $\Gamma_i$ and has larger norm. Hence, every point $A\in\Delta_i$ satisfies
$$
\langle A,u_i\rangle\geq\norm{\Gamma_i},
$$
and therefore $\norm A\geq\norm{\Gamma_i}$. This proves the claimed upper bound.
\qed \end{proof}


Next, we propose a specific input sequence to Algorithm~\ref{algo: bouncing} and derive the induced search cost for $\mathcal C_\phi\text{-}\msp_1(1)$.

\begin{lemma}
\label{lem: choice of sequence g}
For any $\phi\in(0,\pi/2)$ and $\theta\in(\phi,\pi/2]$, let
\begin{align*}
\alpha&:=\frac{\sinn{\phi}}{\sinn{\theta}},\\
\beta&:=\frac{\sinn{\theta}}{\sinn{\theta-\phi}}.
\end{align*}
Then $q_i=\beta^i$, $i\in\integers$, is a valid input to Algorithm~\ref{algo: bouncing}, and the resulting search cost is
$$
\frac{\alpha\beta^3}{\beta-1}.
$$
\end{lemma}

\begin{proof}
Since $\theta\in(\phi,\pi/2]$, we have
$$
\sin\theta-\sin(\theta-\phi)
=2\cos(\theta-\phi/2)\sin(\phi/2)>0,
$$
so $\beta>1$. 
The input condition is satisfied because $\sum_{i=-\infty}^{0}\beta^i=\beta/(\beta-1)<\infty$.

Set $q_0=1$ and $\Gamma_0=(1,0)$. For each $i$, require that $\Gamma_i$ and $\Gamma_{i+1}$ lie on different extreme rays and that the angle at $\Gamma_i$ in the triangle $O\Gamma_i\Gamma_{i+1}$ is $\pi-\theta$. Let $d_i=\norm{\Gamma_i-\Gamma_{i-1}}$. The Law of Sines gives
$$
\frac{d_{i+1}}{\sin\phi}
=
\frac{q_{i+1}}{\sin\theta}
$$
and
$$
\frac{q_{i+1}}{\sin\theta}
=
\frac{q_i}{\sin(\theta-\phi)}.
$$
Hence, $q_{i+1}=\beta q_i$ and $d_{i+1}=\alpha q_{i+1}$, so $q_i=\beta^i$ and $d_i=\alpha\beta^i$.

The angle between the outward ray at $\Gamma_i$ and the segment $\Gamma_i\Gamma_{i+1}$ is $\theta\leq\pi/2$. Therefore,
$$
\left\langle\Gamma_{i+1}-\Gamma_i,\Gamma_i\right\rangle\geq0,
$$
and Lemma~\ref{lem: general bouncing performance} applies and gives the upper bound
$$
\sup_{i\in\integers}
\frac{1}{\beta^i}\sum_{j=-\infty}^{i+2}\alpha\beta^j
=
\frac{\alpha\beta^3}{\beta-1}.
$$
For every even $i$, the point $\Gamma_i$ lies on $\mathcal L_0$, which belongs to $\mathcal C_\phi$. Targets on $\mathcal L_0$ with radii decreasing to $q_i$ from above are first exposed when the agent reaches $\Gamma_{i+2}$. Since the displayed expression is independent of $i$, the upper bound is attained as a supremum.
\qed \end{proof}


We now summarize our findings with the proof of Theorem~\ref{thm: cone search upper bound}.

\begin{proof}[of Theorem~\ref{thm: cone search upper bound}]
For $\theta\in(\phi,\pi/2]$, Lemma~\ref{lem: choice of sequence g} gives search cost
\begin{align*}
\frac{\sin\phi}{\sin\theta}\frac{\beta^3}{\beta-1}
&=
\frac{\sin\phi}{2\sin(\phi/2)}
\left(\frac{\sin\theta}{\sin(\theta-\phi)}\right)^2
\frac{1}{\cos(\theta-\phi/2)}\\
&=
\cos(\phi/2)
\left(\frac{\sin\theta}{\sin(\theta-\phi)}\right)^2
\frac{1}{\cos(\theta-\phi/2)}.
\end{align*}
The function is continuous on $(\phi,\pi/2]$ and tends to infinity as $\theta\rightarrow\phi^+$. Hence, it attains a minimum on this interval.
\qed \end{proof}

\section{Conic Complement Search with a Unit-Speed Agent}
\label{sec: Conic Complement Search}

In this section, we design a trajectory for searching a point in the angle-$\phi$ conic complement, where $\phi \in (0,\pi)$, using a unit-speed agent. That is, we propose a solution to the problem $\mathcal W_\phi\text{-}\msp_1(1)$, and subsequently prove Theorem~\ref{thm: wedge search upper bound}.

To describe the trajectory for searching $\mathcal W_\phi$, we first develop some intuition for why this can be done more efficiently than searching the entire plane $\mathcal P$ with a unit-speed agent, even when $\phi \in (0,\pi)$ is small.
A key aspect in analyzing Algorithm~\ref{algo: uniform spiral trajectory} for $\mathcal P\text{-}\msp_1(1)$ is that the search cost is \emph{independent of the target's angular coordinate}. For this reason, and without loss of generality, we fix a reference ray $\mathcal L_0$ and perform worst-case analysis under the assumption that the target lies on $\mathcal L_0$.

Let $f_0^k(t)$ denote the trajectory of Algorithm~\ref{algo: uniform spiral trajectory}, applicable to $\mathcal P\text{-}\msp_1(1)$. This trajectory intersects $\mathcal L_0$ at the points $X_s = f_0^k(2\pi s)$ and intersects the ray $\mathcal L_\phi$ at the points $Z_s = f_0^k(2\pi s + \phi)$, where $s \in \integers$, see also Figure~\ref{fig:shortcut}. 
To adapt $f_0^k(t)$ for searching $\mathcal W_\phi$ (i.e., the plane excluding the cone $\mathcal C_\phi$), we can modify the trajectory by replacing the curved segment between $X_s$ and $Z_s$ with the shorter straight line segment $X_sZ_s$, while preserving the original trajectory between $Z_s$ and $X_{s+1}$. With this modification, every point in $\mathcal W_\phi$ is reached strictly faster. However, the search cost becomes \emph{dependent} on the target's location. Since the search cost is defined as the supremum over all possible target positions in the unbounded domain $\mathcal W_\phi$, this modification does not necessarily reduce the worst-case cost.

In the following, we explore this idea in more detail by properly defining the sequences of points $X_s$ and $Y_s$, ensuring that the search cost remains independent of the target placement. This, in turn, simplifies the analysis and allows us to demonstrate a strict improvement over searching the full plane.

\begin{figure}[h!]
    \centering
    \includegraphics[width=0.46\textwidth]{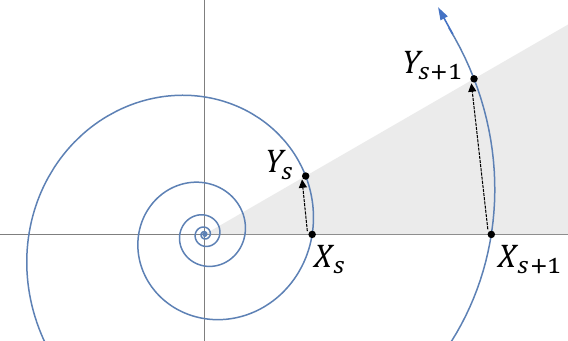}
    \caption{A depiction of the Shortcut Spiral Trajectory to $\mathcal W_\phi\text{-}\msp_1(1)$. The shaded cone corresponds to $\mathcal C_\phi$ with extreme rays $\mathcal L_0,\mathcal L_\phi$, while the shortcut is represented by the dashed arrow within the cone.  
    }
    \label{fig:shortcut}
\end{figure}

\begin{algorithm}[H]
\caption{Shortcut Spiral Trajectory to $\mathcal W_\phi\text{-}\msp_1(1)$}
\label{algo: shortcut spiral}
\begin{algorithmic}
\REQUIRE $k>0$, $1<\lambda <e^{2\pi k}$.
\STATE 
For $s\in \integers$, set 
$X_s := \langle e^{2\pi k s}, 0 \rangle$, and
$Y_s := \langle \lambda e^{2\pi k s}, \phi \rangle$. 
\STATE
Set $\ell := \ln(\lambda)$,
$\alpha := \tfrac{2\pi k-\ell}{2\pi-\phi}$, and
$\beta_s := 2\pi k s + \ell -\alpha \phi$. 
\STATE
Set $h_s(t) := \alpha t + \beta_s$. 
\STATE \textbf{Output Trajectory:}
Line segment $X_s \rightarrow Y_s$, followed by 
$p_s(t) = \langle e^{h_s(t)},t \rangle$ for $t\in [\phi,2\pi]$, and $s\in \integers$
\end{algorithmic}
\end{algorithm}

As before, the conditions $k > 0$ and $1 < \lambda < e^{2\pi k}$ ensure that the trajectory exposes all points in $\mathcal W_\phi$. In particular, the inequality $\lambda < e^{2\pi k}$ implies that $\alpha > 0$. Choosing $\lambda = e^{k \phi}$ recovers the original description of the cross-cut algorithm, where the spiral follows the trajectory of Algorithm~\ref{algo: uniform spiral trajectory} while remaining within $\mathcal W_\phi$.

\begin{lemma}
\label{lem: correctness of shortcut}
For every $\phi \in (0,\pi)$, Algorithm~\ref{algo: shortcut spiral} provides a feasible solution to $\mathcal W_\phi\text{-}\msp_1(1)$.
\end{lemma}

\begin{proof}
We begin by showing that the trajectory is continuous. The trajectory is defined inductively over $s \in \integers$. For each fixed $s$, it moves from $X_s \in \mathcal L_0$ to $Y_s \in \mathcal L_\phi$, followed by a spiral segment $p_s(t)$ for $t \in [\phi, 2\pi]$, see also Figure~\ref{fig:shortcut}. 
Observe that
\begin{align*}
h_s(\phi) &= \alpha \phi + \beta_s 
= \alpha \phi + 2\pi k s + \ell - \alpha \phi 
= 2\pi k s + \ell, \\
h_s(2\pi) &= \alpha 2\pi + \beta_s 
= \alpha (2\pi - \phi) + 2\pi k s + \ell 
= 2\pi k (s+1).
\end{align*}
This implies
$$
\norm{p_s(\phi)} = e^{h_s(\phi)} = \lambda e^{2\pi k s} = \norm{Y_s}, \hspace{0.5em} 
\norm{p_s(2\pi)} = e^{h_s(2\pi)} = e^{2\pi k(s+1)} = \norm{X_{s+1}}.
$$
Hence, $p_s(\phi) = Y_s$ and $p_s(2\pi) = X_{s+1}$, proving that the trajectory is continuous.

Next, we verify feasibility. Each complete piece from $X_s$ through $Y_s$ to $X_{s+1}$ is a scaled copy, by $e^{2\pi ks}$, of the piece with $s=0$ whose length is finite. Since $k>0$, the initial lengths sum to a convergent geometric series and tend to zero towards the origin. Thus the trajectory admits a unit-speed execution from physical time zero. It remains to show that every point in $\mathcal W_\phi$ is eventually exposed.

Let $A=\langle x,\psi\rangle$ with $\psi\in[\phi,2\pi)$ and $x>0$. If $\psi=\phi$, then $\norm{Y_s}=\lambda e^{2\pi ks}$ becomes arbitrarily large, so some $Y_s$ exposes $A$. If $\psi\in(\phi,2\pi)$, then
$$
h_s(\psi)=\alpha\psi+\beta_s=\alpha\psi+2\pi ks+\ell-\alpha\phi
$$
becomes arbitrarily large as $s\to\infty$. Hence, for sufficiently large $s$, $x<e^{h_s(\psi)}$, and $p_s(\psi)$ exposes $A$.
\qed \end{proof}

It follows that for each $s \in \integers$, the set of targets exposed by $p_s(t)$, for $t \in [\phi, 2\pi)$, forms a nested family in $s$, and their union covers $\mathcal W_\phi$. Define the \emph{wedge} $W_s \subseteq \mathcal W_\phi$ as the set of targets first exposed by $p_{s+1}(t)$ for some $t \in [\phi,2\pi)$. The wedges $\{W_s\}_s$ are disjoint, and by Lemma~\ref{lem: correctness of shortcut}, we have $\bigcup_s W_s = \mathcal W_\phi$.

We now analyze the search cost of Algorithm~\ref{algo: shortcut spiral}. The following observation will be useful.

\begin{lemma}
\label{lem: compute segments}
For every $s \in \integers$ and $\phi\leq t_1\leq t_2\leq2\pi$, we have
$$
\int_{t_1}^{t_2} \norm{p_s'(t)} \dd t 
= \frac{\sqrt{1+\alpha^2}}{\alpha} e^{2\pi k s} \left( e^{h_0(t_2)} - e^{h_0(t_1)} \right),
$$ 
 and 
$$\norm{X_s Y_s} = e^{2\pi k s} \norm{X_0 Y_0},
$$
where $\alpha$, $p_s(t)$, $X_s$, and $Y_s$ are as defined in Algorithm~\ref{algo: shortcut spiral}.
\end{lemma}

\begin{proof}
By Lemma~\ref{lem: standard integral}, and since $\alpha > 0$,
$$
\int_{t_1}^{t_2} \norm{p_s'(t)} \dd t
= \frac{\sqrt{1+\alpha^2}}{\alpha} \left( e^{h_s(t_2)} - e^{h_s(t_1)} \right)
= \frac{\sqrt{1+\alpha^2}}{\alpha} e^{2\pi k s} \left( e^{h_0(t_2)} - e^{h_0(t_1)} \right),
$$
using the identity $h_s(t) = 2\pi k s + h_0(t)$.
Also, in Cartesian coordinates, $X_s = e^{2\pi k s}(1,0)$ and 
$$Y_s = \lambda e^{2\pi k s}(\coss{\phi}, \sinn{\phi}),$$
so
$$
\norm{X_s Y_s} = e^{2\pi k s} \norm{ (1,0) - \lambda (\coss{\phi}, \sinn{\phi}) } = e^{2\pi k s} \norm{X_0 Y_0}.
$$
\qed \end{proof}

We now compute the time required to expose a point on $\mathcal L_\psi$, where $\psi \in [\phi, 2\pi]$, which guides the optimal choice of $k$ and $\lambda$ to minimize the search cost.

\begin{lemma}
\label{lem: exposure time}
For every $\psi \in [\phi,2\pi)$ and $s \in \integers$, any target in $W_s \cap \mathcal L_\psi$ is exposed by $p_{s+1}(\psi)$ at time
$$
e^{2\pi k(s+1)} \left(
\frac{\sqrt{1+\alpha^2}}{\alpha} \lambda e^{\alpha(\psi - \phi)} 
+ \frac{e^{2\pi k}}{e^{2\pi k} - 1} \left( \norm{X_0 Y_0} - \frac{\sqrt{1+\alpha^2}}{\alpha}(\lambda - 1) \right)
\right).
$$
\end{lemma}

\begin{proof}
We have $p_{s+1}(\psi) \in \mathcal L_\psi$ with 
$
\norm{p_{s+1}(\psi)} = e^{h_{s+1}(\psi)} = e^{2\pi k} e^{h_s(\psi)},
$
and by definition of $W_s$, this is at least $\norm{A}$ for all $A \in W_s \cap \mathcal L_\psi$.
By Lemma~\ref{lem: compute segments}, the time to reach $p_{s+1}(\psi)$ is
\begin{align*}
&\sum_{i=-\infty}^{s+1} \norm{X_i Y_i}
+ \sum_{i=-\infty}^{s} \int_\phi^{2\pi} \norm{p_i'(t)} \dd t 
+ \int_\phi^{\psi} \norm{p_{s+1}'(t)} \dd t \\
&= \norm{X_0 Y_0} \sum_{i=-\infty}^{s+1} e^{2\pi k i} 
+ \frac{\sqrt{1+\alpha^2}}{\alpha} \left( 
\left( e^{h_0(2\pi)} - e^{h_0(\phi)} \right) \sum_{i=-\infty}^{s} e^{2\pi k i}
+ e^{2\pi k(s+1)} \left( e^{h_0(\psi)} - e^{h_0(\phi)} \right) 
\right).
\end{align*}

Using $h_0(t) = \alpha(t - \phi) + \ell$, $h_0(\phi) = \ell$, and $h_0(2\pi) = 2\pi k$, and simplifying gives the claimed expression.
\qed \end{proof}

We now show that we can choose $k$ and $\lambda$ to make the search cost independent of the target's location.

\begin{lemma}
\label{lem: choice of k lambda for Wphi}
For $\phi \in (0,\pi)$ and any $\lambda > 1$, define
$$
k_\lambda := \frac{\ln \lambda}{2\pi} + \frac{\lambda - 1}{\sqrt{2\lambda(1 - \coss{\phi})}} \left( 1 - \frac{\phi}{2\pi} \right),
$$ 
and 
$$
\alpha_\lambda := \frac{2\pi k_\lambda - \ln \lambda}{2\pi - \phi}.
$$
Then, $(k_\lambda, \lambda)$ is a valid input to Algorithm~\ref{algo: shortcut spiral} for $\mathcal W_\phi\text{-}\msp_1(1)$, and the search cost equals
$$
\frac{\sqrt{1 + \alpha_\lambda^2}}{\alpha_\lambda} e^{2\pi k_\lambda}.
$$
\end{lemma}


\begin{proof}
Since $\lambda>1$,
$$
k_\lambda>\frac{\ln\lambda}{2\pi}>0,
$$
so $1<\lambda<e^{2\pi k_\lambda}$.

Let $C=\sqrt{1+\alpha^2}/\alpha$ and $D=\norm{X_0Y_0}$. Lemma~\ref{lem: exposure time} and the lower bound
$$
\norm A>\norm{p_s(\psi)}=\lambda e^{2\pi ks}e^{\alpha(\psi-\phi)}
$$
show that the search cost is at most
$$
Ce^{2\pi k}
+
\frac{e^{4\pi k}}{\lambda(e^{2\pi k}-1)e^{\alpha(\psi-\phi)}}
\left(D-C(\lambda-1)\right).
$$
Now
$$
D=\sqrt{\lambda^2-2\lambda\cos\phi+1}.
$$
The equation $D=C(\lambda-1)$ is equivalent to
$$
\alpha=\frac{\lambda-1}{\sqrt{2\lambda(1-\cos\phi)}}.
$$
Together with $\alpha=(2\pi k-\ln\lambda)/(2\pi-\phi)$, this gives exactly $k=k_\lambda$. For this choice, the second term vanishes and the search cost is at most
$$
\frac{\sqrt{1+\alpha_\lambda^2}}{\alpha_\lambda}e^{2\pi k_\lambda}.
$$
Targets whose radii decrease to $\norm{p_s(\psi)}$ from above show that this bound is attained as a supremum.
\qed \end{proof}


We are now ready to prove Theorem~\ref{thm: wedge search upper bound}, which follows from Lemma~\ref{lem: choice of k lambda for Wphi} by showing that the search cost admits a minimum for every $\phi \in (0,\pi)$ when the parameter $\lambda$ is restricted to values strictly greater than $1$.

\begin{proof}[of Theorem~\ref{thm: wedge search upper bound}]
By Lemma~\ref{lem: choice of k lambda for Wphi}, the search cost for $\lambda>1$ is
$$
p_\phi(\lambda)
=
\frac{\sqrt{\lambda^2-2\lambda\cos\phi+1}}{\lambda-1}e^{2\pi k_\lambda}.
$$
This function is continuous on $(1,\infty)$, and
$$
\lim_{\lambda\rightarrow1^+}p_\phi(\lambda)
=
\lim_{\lambda\to\infty}p_\phi(\lambda)
=
\infty.
$$
Therefore, it attains a minimum $\mathcal R_\phi$ at some $\lambda>1$.

It remains to prove $\mathcal R_\phi<\mathcal U_1$. Let $\kappa_1$ be the optimizer from Lemma~\ref{lem: closed form of k_i}. The function $k_\lambda$ is strictly increasing in $\lambda$, since both $\ln\lambda$ and
$$
\frac{\lambda-1}{\sqrt\lambda}=\sqrt\lambda-\frac1{\sqrt\lambda}
$$
are strictly increasing. Since $\lim_{\lambda\rightarrow1^+}k_\lambda=0$, set $\lambda_0=e^{\kappa_1\phi}$. Then
\begin{align*}
k_{\lambda_0}-\kappa_1
&=
\left(1-\frac{\phi}{2\pi}\right)
\left(
\frac{\sinh(\kappa_1\phi/2)}{\sin(\phi/2)}-\kappa_1
\right)>0,
\end{align*}
where the strict inequality follows from $\sinh x>x$ and $\sin y<y$ for positive $x,y$. Hence, there is a unique $\lambda_*\in(1,\lambda_0)$ such that $k_{\lambda_*}=\kappa_1$. For this value,
$$
\alpha_*
=
\frac{2\pi\kappa_1-\ln\lambda_*}{2\pi-\phi}
>
\kappa_1.
$$
Since $\sqrt{1+x^2}/x$ is strictly decreasing,
$$
\mathcal R_\phi
\leq
p_\phi(\lambda_*)
=
\frac{\sqrt{1+\alpha_*^2}}{\alpha_*}e^{2\pi\kappa_1}
<
\frac{\sqrt{1+\kappa_1^2}}{\kappa_1}e^{2\pi\kappa_1}
=
\mathcal U_1.
$$
\qed \end{proof}

Theorem~\ref{thm: wedge search upper bound} will be used in the next section to establish the performance of the Cone-Wedge Hybrid Algorithm~\ref{algo: cone wedge hybrid} for $\mathcal P\textrm{-}\msp_2(1,c)$. 
Algorithm~\ref{algo: shortcut spiral} will be utilized by the speed-$1$ agent. 
Recall that for each $\phi$, Algorithm~\ref{algo: shortcut spiral} is parameterized by $\lambda, k$, whereas the optimal parameters are set to $\lambda=\lambda(\phi)$ and $k_\lambda=k_\lambda(\phi)$ by Theorem~\ref{thm: wedge search upper bound}. 
At the same time, Algorithm~\ref{algo: shortcut spiral} emerged by modifying Algorithm~\ref{algo: uniform spiral trajectory} which corresponds to choosing $\lambda=e^{k \phi}$. We are therefore motivated to explore for each $\phi$ how $e^{k_\lambda \phi}$ compares to $\lambda$. This is done in Figure~\ref{fig:actualk_lambdaexpcomp}.

\begin{figure}[h!]
    \centering
    \begin{subfigure}{0.44\textwidth}
        \centering
        \includegraphics[width=\linewidth]{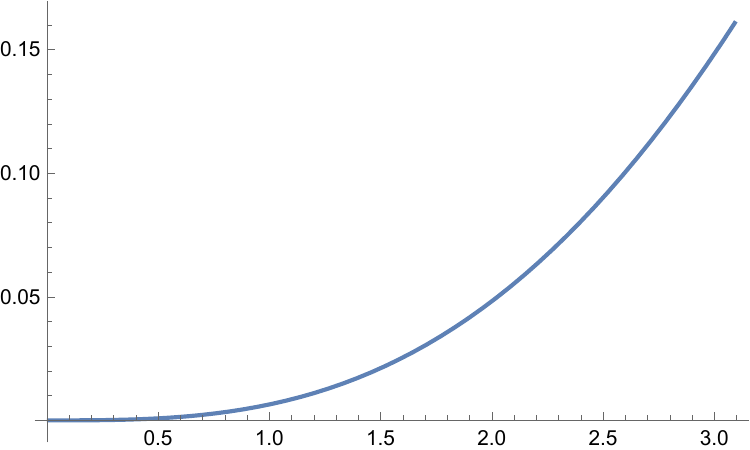}
        \caption{
        Comparison between $e^{k_\lambda \phi}$ and $\lambda$, for the optimal values of $\lambda(\phi),k_\lambda(\phi)$ as functions of $\phi$. The plot shows $e^{k_\lambda \phi} - \lambda$, as a function of $\phi$. 
        }
        \label{fig:lambdaexpcomp}
    \end{subfigure}
    \hfill
    \begin{subfigure}{0.44\textwidth}
        \centering
        \includegraphics[width=\linewidth]{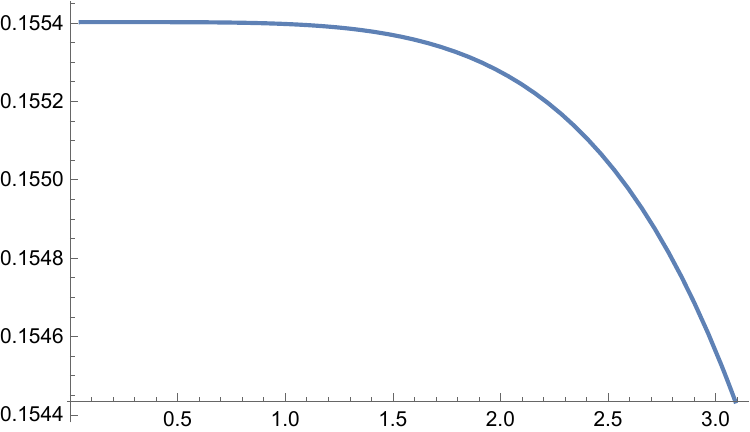}
        \caption{The plot of the optimal $k_\lambda=k_\lambda(\phi)$ 
as it is used by Algorithm~\ref{algo: shortcut spiral}.
        }
        \label{fig:actualk}
    \end{subfigure}
    \caption{
    Plot of parameters $\lambda=\lambda(\phi)$ and $k_\lambda=k_\lambda(\phi)$ as computed by Theorem~\ref{thm: wedge search upper bound} and used in 
    Algorithm~\ref{algo: shortcut spiral} for searching $\mathcal W_\phi$, $\phi \in (0,\pi)$. 
    The horizontal axis in both plots corresponds to the angular parameter $\phi$. 
    }
    \label{fig:actualk_lambdaexpcomp}
\end{figure}

\section{Sub-Optimality of the Spiral-Type Off-Set Algorithm}
\label{sec: suboptimality}

In this section, we  demonstrate that 
the minimum search cost within the common-expansion-factor Off-Set family of Algorithm~\ref{algo: Off-Set Trajectory}, as established in Theorem~\ref{thm: simplified general spiral upper bound}, can be strictly improved when the speed of one of the agents is sufficiently small. We illustrate this by focusing on the special case of two agents searching the plane, that is, the problem $\mathcal P\text{-}\msp_2(1,c)$, where $c \in (0,1)$. This establishes Theorems~\ref{thm: hybrid upper bound} and~\ref{thm: numerical upper bound}.
We begin with an observation. 

\begin{proposition}
\label{prop: no slow agent}
The optimal search cost for $\mathcal P\text{-}\msp_2(1,c)$ with $c \leq 1/\mathcal U_1 \approx 0.05783906146$ is $\mathcal U_1 \approx 17.28935$.
\end{proposition}

\begin{proof}
The unit-speed agent can follow an optimal one-agent trajectory, so the optimal search cost is at most $\mathcal U_1$.

Consider any feasible two-agent algorithm. If the speed-$c$ agent is the first to expose some target $A$, then it must travel length at least $\norm A$, and therefore
$$
\frac{t_A}{\norm A}\geq\frac1c\geq\mathcal U_1.
$$
If the speed-$c$ agent is never the first to expose a target, then the unit-speed agent exposes every target by itself, and the one-agent lower bound gives search cost at least $\mathcal U_1$. In either case, the search cost is at least $\mathcal U_1$.
\qed \end{proof}


Proposition~\ref{prop: no slow agent} is a negative result, as it shows that slow agents with speed $c < 1/\mathcal U_1$ cannot help reduce the optimal search cost below $\mathcal U_1$ in the problem $\mathcal P\text{-}\msp_2(1,c)$. However, can the Off-Set logarithmic-spiral Algorithm~\ref{algo: Off-Set Trajectory} make use of such slow-moving agents? The following immediate corollary of Theorem~\ref{thm: simplified general spiral upper bound} addresses this.

\begin{corollary}
\label{cor: 2 speed upper bound corollary}
The upper bound obtained from the Off-Set construction~\ref{algo: Off-Set Trajectory} for $\mathcal P\text{-}\msp_2(1,c)$ equals
$\mathcal U_1$ if $c \leq (\mathcal U_2/\mathcal U_1)^2 \approx 0.2666865601$, and
$\mathcal U_2/\sqrt{c}$ if $c > (\mathcal U_2/\mathcal U_1)^2$.
\end{corollary}

Next, we formally demonstrate how to strictly improve the performance of the Off-Set Algorithm~\ref{algo: Off-Set Trajectory} for $\mathcal P\text{-}\msp_2(1,c)$ when  
$$
1/\mathcal U_1 < c < \left(\mathcal U_2/ \mathcal U_1\right)^2.
$$
Moreover, this improvement extends to speeds slightly exceeding the threshold $\left(\mathcal U_2/\mathcal U_1\right)^2$, at which the two-agent Off-Set bound first becomes smaller than $\mathcal U_1$. In this extended range, the speed-$c$ agent exposes new points while incurring a strictly lower search cost than $\mathcal U_1$, though further improvements are still possible.

In the minimizing common-expansion-factor Off-Set construction, every participating agent covers all directions. The restricted-family lower bound proves the dichotomy in Corollary~\ref{cor: 2 speed upper bound corollary}: either the one-agent choice is sufficient, or the minimum is $\mathcal U_2/\sqrt{c}$. At the activation threshold the choices tie. The hybrid strategy instead partitions the plane into two geometrically different regions. The speed-$c$ agent searches a cone, while the unit-speed agent searches the conic complement. This division yields an improvement over the common-expansion-factor Off-Set family, not a comparison with all possible spiral or all-direction strategies.

We begin by describing an alternative algorithm for $\mathcal P\text{-}\msp_2(1,c)$.

\begin{algorithm}[H]
\caption{Cone-Wedge Hybrid Algorithm for $\mathcal P\text{-}\msp_2(1,c)$}
\label{algo: cone wedge hybrid}
\begin{algorithmic}
\REQUIRE $\phi\in(0,\pi/2)$ with $\gamma(\phi)\leq1$, and $c=\gamma(\phi)$
\STATE Let $\theta_\phi$ be a value attaining the minimum $\mathcal F_\phi$ in Theorem~\ref{thm: cone search upper bound}, and set
$
q_i=\left(\frac{\sin\theta_\phi}{\sin(\theta_\phi-\phi)}\right)^i,
$
for $i\in\integers$.
\STATE Let $\lambda_\phi>1$ be a value attaining the minimum $\mathcal R_\phi$ in Theorem~\ref{thm: wedge search upper bound}, and set $k=k_{\lambda_\phi}$ as in Lemma~\ref{lem: choice of k lambda for Wphi}.
\STATE \textbf{Output Trajectories:} \\
The speed-$c$ agent searches $\mathcal C_\phi$ using Algorithm~\ref{algo: bouncing} with input sequence $\{q_i\}_{i \in \integers}$.\\
The unit-speed agent searches $\mathcal W_\phi$ using Algorithm~\ref{algo: shortcut spiral} with parameters $\lambda_\phi$ and $k_{\lambda_\phi}$.
\end{algorithmic}
\end{algorithm}

We now demonstrate the effectiveness of Algorithm~\ref{algo: cone wedge hybrid} by proving Theorem~\ref{thm: hybrid upper bound}.

\begin{proof}[of Theorem~\ref{thm: hybrid upper bound}]
We analyze the performance of Algorithm~\ref{algo: cone wedge hybrid} for input $\phi\in(0,\pi/2)$ with $\gamma(\phi)\leq1$.  
The plane is partitioned into the cone $\mathcal{C}_\phi$, searched by the speed-$\gamma(\phi)$ agent, and the wedge $\mathcal{W}_\phi$, searched by the unit-speed agent.  
By Theorem~\ref{thm: cone search upper bound}, the cost of the cone agent alone, restricted to $\mathcal{C}_\phi$, is
$$
\coss{\phi/2} \cdot \frac{\mathcal{F}_\phi}{\gamma(\phi)},
$$
i.e., the unit-speed cost scaled by $1/\gamma(\phi)$.  
Similarly, Theorem~\ref{thm: wedge search upper bound} gives a search cost of $\mathcal{R}_\phi$ for points in $\mathcal{W}_\phi$.  
By the definition of $\gamma(\phi)$, these expressions are equal, so the overall search cost is at most $\mathcal R_\phi$. To see equality, fix an interior wedge ray $\mathcal L_\psi$, with $\phi<\psi<2\pi$. The cone agent never exposes a target on this ray. For any $s$, targets on this ray with radii decreasing to $\norm{p_s(\psi)}$ from above are first exposed by the fast agent at $p_{s+1}(\psi)$. Their normalized discovery times tend to $\mathcal R_\phi$, by Lemma~\ref{lem: choice of k lambda for Wphi}. Thus the overall supremum equals $\mathcal R_\phi$.
\qed \end{proof}

\begin{proof}[of Theorem~\ref{thm: numerical upper bound}]
Write $C(\phi)=\cos(\phi/2)\mathcal F_\phi$ and $R(\phi)=\mathcal R_\phi$ for $0<\phi<\pi/2$. We first prove that $C(\phi)R(\phi)$ and $\gamma(\phi)=C(\phi)/R(\phi)$ are strictly increasing. 
This will establish a unique transition value and the ordering of the two bounds without relying on numerical sampling.

By Lemma~\ref{lem: properties of C phi} in Appendix~\ref{app: numerical upper bound}, $C(\phi)$ is differentiable and strictly increasing for $0<\phi<\pi/2$, and
$$
\frac{C'(\phi)}{C(\phi)}
\geq
\frac{2}{\sqrt5}.
$$
Moreover,
$
\lim_{\phi\rightarrow0^+}C(\phi)=1,
$
and 
$\lim_{\phi\rightarrow(\pi/2)^-}C(\phi)=\infty.
$

By Lemma~\ref{lem: properties of R phi} in Appendix~\ref{app: numerical upper bound}, $R(\phi)$ is continuous and strictly decreasing, $R(\phi)>1$, and
$
\lim_{\phi\rightarrow0^+}R(\phi)=\mathcal U_1.
$
Moreover, for every $0\leq\phi_1<\phi_2<\pi/2$,
$$
e^{-2(\phi_2-\phi_1)/(3\pi)}R(\phi_1)
<
R(\phi_2)
<
R(\phi_1).
$$

Combining the above inequalities with $C'(\phi)/C(\phi)\geq2/\sqrt5$, for $0<\phi_1<\phi_2<\pi/2$ we get
$$
\frac{C(\phi_2)R(\phi_2)}{C(\phi_1)R(\phi_1)}
>
\exp\left(
\left(\frac2{\sqrt5}-\frac2{3\pi}\right)(\phi_2-\phi_1)
\right)>1.
$$
Therefore, the function $C(\phi)R(\phi)$ is strictly increasing. Since $C(\phi)$ is strictly increasing and $R(\phi)$ is strictly decreasing, the function
$$
\gamma(\phi)=\frac{C(\phi)}{R(\phi)}
$$
is also strictly increasing.
The limits already established give
\begin{align*}
\lim_{\phi\rightarrow0^+}&\gamma(\phi)
=1/\mathcal U_1, \\
\lim_{\phi\rightarrow(\pi/2)^-}&\gamma(\phi)
 =\infty,\\
\lim_{\phi\rightarrow0^+}&C(\phi)R(\phi)
=\mathcal U_1, \\
\lim_{\phi\rightarrow(\pi/2)^-}&C(\phi)R(\phi)
=\infty.
\end{align*}
Moreover,
$
\mathcal U_1\leq e\sqrt{1+4\pi^2}<\pi^2e^2\leq\mathcal U_2^2.
$
Here the first inequality evaluates the unit-speed objective at $k=1/(2\pi)$, and the last uses
$\mathcal U_2\geq\min_{k>0}e^{\pi k}/k=\pi e$.
Thus there is exactly one $\phi_*\in(0,\pi/2)$ with $C(\phi_*)R(\phi_*)=\mathcal U_2^2$.

For every $c>1/\mathcal U_1$, let $\phi(c)$ be the unique angle satisfying $\gamma(\phi(c))=c$. 
Whenever $c\leq1$, the corresponding angle $\phi(c)$ can be used in the hybrid algorithm.
Since $c=C(\phi)/R(\phi)$,
\begin{align*}
R(\phi)<\frac{\mathcal U_2}{\sqrt c}
&\quad\Leftrightarrow\quad C(\phi)R(\phi)<\mathcal U_2^2
\quad\Leftrightarrow\quad\phi<\phi_* 
\quad\Leftrightarrow\quad c<c_*.
\end{align*}
Equality holds for $c=c_*$, while for $c>c_*$ the bound of Theorem~\ref{thm: simplified general spiral upper bound} is smaller. Moreover, by Theorem~\ref{thm: wedge search upper bound},
$
R(\phi)<\mathcal U_1.
$
In particular, at $\phi=\phi_*$,
$$
c_*
=
\frac{\mathcal U_2^2}{R(\phi_*)^2}
>
\left(\frac{\mathcal U_2}{\mathcal U_1}\right)^2.
$$
Therefore, the interval in which the hybrid algorithm improves the bound of Corollary~\ref{cor: 2 speed upper bound corollary} extends strictly beyond the value $(\mathcal U_2/\mathcal U_1)^2$. Together with Proposition~\ref{prop: no slow agent}, this proves the claimed upper bound.

We conclude by numerically evaluating the unique solution $\phi_*$ of
$
C(\phi)R(\phi)=\mathcal U_2^2.
$
This gives
$
\phi_*=0.7510108432\ldots.
$
Since $\gamma(\phi)$ is strictly increasing,
$
c_*=\gamma(\phi_*)=0.2681784154\ldots.
$
Moreover,
$
R(\phi_*)=17.2411971300\ldots.
$
\qed \end{proof}

\begin{figure}[h!]
    \centering
    \includegraphics[width=0.45\textwidth]{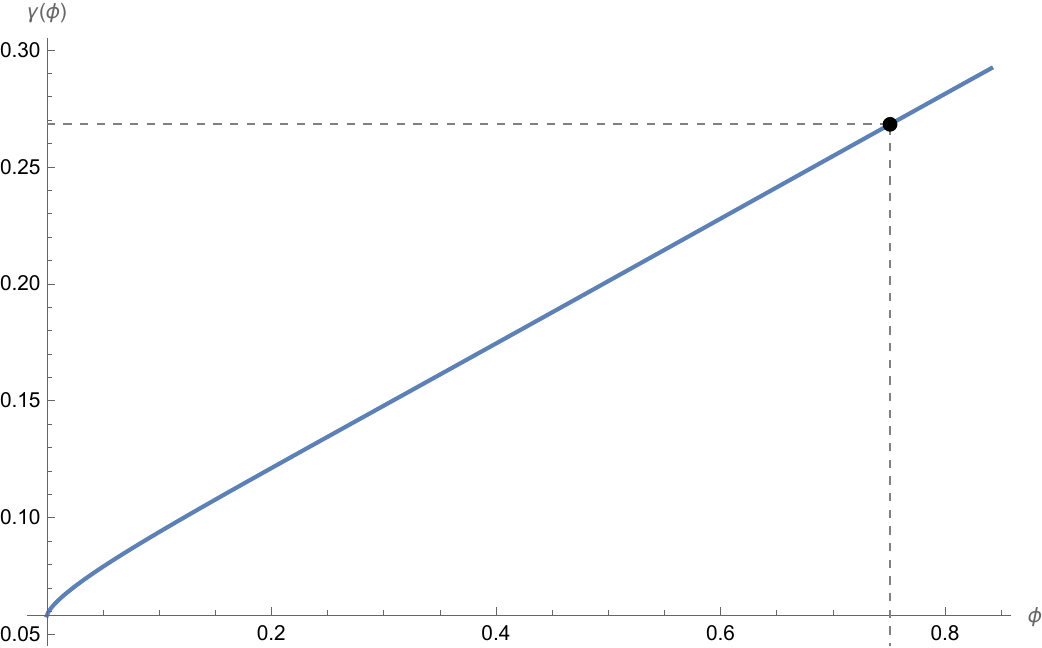}
\caption{
Plot of $\gamma(\phi)=\coss{\phi/2}\tfrac{\mathcal F_\phi}{\mathcal R_\phi}$ for $0<\phi\leq\phi_{\mathrm c}$, where $\phi_{\mathrm c}=\arccos(2/3)$, and where $\mathcal F_\phi$ and $\mathcal R_\phi$ are defined in Theorems~\ref{thm: cone search upper bound} and~\ref{thm: wedge search upper bound}, respectively. 
This is the range containing the transition value $\phi_*$ used in Theorem~\ref{thm: numerical upper bound}.
The dotted lines indicate the transition point
$(\phi_*,c_*)$, where
$c_*=\gamma(\phi_*)$.
}
    \label{fig: range of gamma(phi)}
\end{figure}

\section{Discussion}

We studied point search in the plane using multi-speed agents, focusing on extensions of spiral-based strategies that are known to be optimal in the single-agent, unit-speed setting. While spirals provide a natural baseline, our results show that the common-expansion-factor Off-Set family need not be optimal once agents have heterogeneous speeds. We identify a concrete speed interval in which the hybrid strategy improves on every member of that family. This leaves open the performance of more general spiral strategies.

The quantitative improvement over the two-active-agent Off-Set bound is modest and confined to a narrow range of speed ratios; the hybrid also improves on the one-agent bound for speeds between $1/\mathcal U_1$ and the Off-Set activation threshold. We do not claim that this hybrid is close to optimal, and we currently lack intuition about the structure of globally optimal strategies in this regime. Rather, it demonstrates that assigning geometrically distinct tasks can outperform the common-expansion-factor Off-Set family.

Determining the optimal search cost for $\mathcal C_\phi$ remains an interesting open problem. For $n$ agents of equal unit speed, we conjecture that spiral-based strategies remain optimal, although establishing this lies outside the scope of the present work. Finally, for sufficiently large speed ratios, we conjecture that the offset spiral strategy analyzed here becomes optimal again, reflecting the diminishing influence of slower agents.

Randomized strategies define a largely orthogonal direction. Randomizing the global orientation of a deterministic trajectory removes the adversary’s ability to select a worst-case target angle, so that the expected performance coincides with the average-case performance of the corresponding deterministic strategy under a uniformly random target direction. A systematic comparison between deterministic and randomized strategies in the multi-speed planar search model remains open.

Several further questions arise naturally. Can matching lower bounds be established for restricted classes of speed distributions? Does a similar separation between spiral-based and non-spiral strategies arise in shoreline search with heterogeneous agents? How should strategies adapt when agent speeds are unknown or time-varying, or when search is conducted in bounded domains or environments with obstacles?

\bibliographystyle{plainurl}
\bibliography{PointSearch2Dbib-new}

@inproceedings{GeorgiouFermatSTACS26,
  author    = {Georgiou, K.},
  title     = {Optimal Average Disk-Inspection via Fermat’s Principle},
  booktitle = {Proceedings of the 43rd International Symposium on Theoretical Aspects of Computer Science (STACS 2026)},
  address   = {Grenoble, France},
  year      = {2026},
  month     = mar,
  note      = {To appear; March 10--13, 2026}
}

@inproceedings{GJMSpiralLatin26,
  author    = {Georgiou, K. and Jones, C. and Madej, M.},
  title     = {Spirals and Beyond: Competitive Plane Search with Multi-Speed Agents},
  booktitle = {Proceedings of the Latin American Symposium on Theoretical Informatics (LATIN 2026)},
  year      = {2026},
  note      = {To appear}
}

@Inbook{Lopez-Ortiz2016,
  author    = "L{\'o}pez-Ortiz, Alejandro",
  editor    = "Kao, Ming-Yang",
  title     = "Alternative Performance Measures in Online Algorithms",
  bookTitle = "Encyclopedia of Algorithms",
  year      = "2016",
  publisher = "Springer New York",
  address   = "New York, NY",
  pages     = "67--72",
  isbn      = "978-1-4939-2864-4",
  doi       = "10.1007/978-1-4939-2864-4_13",
  url       = "https://doi.org/10.1007/978-1-4939-2864-4_13"
}

@article{czyzowicz2014beachcombers,
  title   = {The Beachcombers' Problem: Walking and Searching with Mobile Robots},
  journal = {Theoretical Computer Science},
  volume  = {608},
  pages   = {201-218},
  year    = {2015},
  author  = {Czyzowicz, J. and Gasieniec, L. and Georgiou, K. and Kranakis, E. and MacQuarrie, F.}
}

@inproceedings{lamprou2016fast,
  title     = {Fast Two-Robot Disk Evacuation with Wireless Communication},
  author    = {Lamprou, I. and Martin, R. and Schewe, S.},
  booktitle = {International Symposium on Distributed Computing (DISC)},
  pages     = {1--15},
  year      = {2016},
  publisher = {Springer}
}

@inproceedings{georgiou2022evacuation,
  title     = {Evacuation from a Disk for Robots with Asymmetric Communication},
  author    = {Georgiou, K. and Giachoudis, N. and Kranakis, E.},
  booktitle = {33rd International Symposium on Algorithms and Computation (ISAAC)},
  pages     = {19--1},
  year      = {2022},
  publisher = {Schloss Dagstuhl--Leibniz-Zentrum f{\"u}r Informatik}
}

@article{georgiou2024overcoming,
  title   = {Overcoming Probabilistic Faults in Disoriented Linear Search},
  author  = {Georgiou, K. and Giachoudis, N. and Kranakis, E.},
  journal = {Theoretical Computer Science},
  volume  = {1014},
  pages   = {114761},
  year    = {2024},
  publisher={Elsevier}
}

@InProceedings{coleman2025multimodal,
  title     = {Multimodal Search on a Line},
  author    = {Coleman, J. and Ivanov, D. and Kranakis, E. and Krizanc, D. and Ponce, O. M.},
  booktitle = {International Colloquium on Structural Information and Communication Complexity (SIROCCO)},
  year      = {2025},
  series    = {Lecture Notes in Computer Science},
  publisher = {Springer}
}

@inproceedings{kranakis2024survey,
  title     = {A Survey of the Impact of Knowledge on the Competitive Ratio in Linear Search},
  author    = {Kranakis, E.},
  booktitle = {International Symposium on Stabilizing, Safety, and Security of Distributed Systems},
  pages     = {23--38},
  year      = {2024},
  publisher = {Springer}
}

@article{chung2011search,
  title   = {Search and Pursuit-Evasion in Mobile Robotics: A Survey},
  author  = {Chung, T. and Hollinger, G. and Isler, V.},
  journal = {Autonomous Robots},
  volume  = {31},
  pages   = {299--316},
  year    = {2011},
  publisher={Springer}
}

@inproceedings{langetepe2010optimality,
  title     = {On the Optimality of Spiral Search},
  author    = {Langetepe, E.},
  booktitle = {Proceedings of the Twenty-First Annual ACM-SIAM Symposium on Discrete Algorithms},
  pages     = {1--12},
  year      = {2010},
  publisher = {SIAM}
}

@article{finch2005logarithmic,
  title   = {The Logarithmic Spiral Conjecture},
  author  = {Finch, S. R.},
  journal = {arXiv preprint math/0501133},
  year    = {2005}
}

@article{bellman1958dynamic,
  title   = {Dynamic Programming},
  author  = {Bellman, R.},
  journal = {Chapter IX, Princeton University Press, Princeton, New Jersey},
  year    = {1958}
}

@article{berzsenyi1995lost,
  title   = {Lost in a Forest (A Problem Area Initiated by the Late Richard E. Bellman)},
  author  = {Berzsenyi, G.},
  journal = {Quantum (November/December, 1995)},
  volume  = {41},
  year    = {1995}
}

@article{kubel2021approximation,
  title   = {On the Approximation of Shortest Escape Paths},
  author  = {K{\"u}bel, D. and Langetepe, E.},
  journal = {Computational Geometry},
  volume  = {93},
  pages   = {101709},
  year    = {2021},
  publisher={Elsevier}
}

@misc{gibbs2016bellman,
  title     = {Bellman’s Escape Problem for Convex Polygons},
  author    = {Gibbs, P.},
  year      = {2016},
  publisher = {viXra}
}

@article{finch2004lost,
  title   = {Lost in a Forest},
  author  = {Finch, S. R. and Wetzel, J. E.},
  journal = {The American Mathematical Monthly},
  volume  = {111},
  number  = {8},
  pages   = {645--654},
  year    = {2004},
  publisher={Taylor \& Francis}
}

@InProceedings{AcharjeeGKS19,
  title     = {Lower Bounds for Shoreline Searching With 2 or More Robots},
  author    = {Acharjee, S. and Georgiou, K. and Kundu, S. and Srinivasan, A.},
  publisher = {Schloss Dagstuhl - LZI},
  year      = {2019},
  booktitle = {23rd OPODIS},
  pages     = {26:1--26:11},
  series    = {LIPIcs},
}

@inproceedings{dobrev2020improved,
  title     = {Improved Lower Bounds for Shoreline Search},
  author    = {Dobrev, S. and Kr{\'a}lovi{\v{c}}, R. and Pardubsk{\'a}, D.},
  booktitle = {International Colloquium on Structural Information and Communication Complexity},
  pages     = {80--90},
  year      = {2020},
  publisher = {Springer}
}

@article{isbell1957optimal,
  title   = {An Optimal Search Pattern},
  author  = {Isbell, J. R.},
  journal = {Naval Research Logistics Quarterly},
  volume  = {4},
  number  = {4},
  pages   = {357--359},
  year    = {1957},
  publisher={Wiley Online Library}
}

@article{gluss1961minimax,
  title   = {The Minimax Path in a Search for a Circle in a Plane},
  author  = {Gluss, B.},
  journal = {Naval Research Logistics Quarterly},
  volume  = {8},
  number  = {4},
  pages   = {357--360},
  year    = {1961},
  publisher={Wiley Online Library}
}

@inproceedings{baeza1988searching,
  title     = {Searching with Uncertainty},
  author    = {Baeza-Yates, R. A. and Culberson, J. C. and Rawlins, G. J. E.},
  booktitle = {Scandinavian Workshop on Algorithm Theory},
  pages     = {176--189},
  year      = {1988},
  publisher = {Springer}
}

@article{baeza1995parallel,
  title   = {Parallel Searching in the Plane},
  author  = {Baeza-Yates, R. and Schott, R.},
  journal = {Computational Geometry},
  volume  = {5},
  number  = {3},
  pages   = {143--154},
  year    = {1995},
  publisher={Elsevier}
}

@article{jez2009two,
  title   = {On the Two-Dimensional Cow Search Problem},
  author  = {Je{\.z}, A. and {\L}opusza{\'n}ski, J.},
  journal = {Information Processing Letters},
  volume  = {109},
  number  = {11},
  pages   = {543--547},
  year    = {2009},
  publisher={Elsevier}
}

@article{gal2010search,
  title   = {Search Games},
  author  = {Gal, S.},
  journal = {Wiley Encyclopedia of Operations Research and Management Science},
  year    = {2010},
  publisher={Wiley Online Library}
}

@article{gluss1961alternative,
  title   = {An Alternative Solution to the “Lost at Sea” Problem},
  author  = {Gluss, B.},
  journal = {Naval Research Logistics Quarterly},
  volume  = {8},
  number  = {1},
  pages   = {117--122},
  year    = {1961},
  publisher={Wiley Online Library}
}

@article{finch2005searching,
  title   = {Searching for a Shoreline},
  author  = {Finch, S. R. and Zhu, L.-Y.},
  journal = {arXiv preprint math/0501123},
  year    = {2005}
}

@article{langetepe2012searching,
  title   = {Searching for an Axis-Parallel Shoreline},
  author  = {Langetepe, E.},
  journal = {Theoretical Computer Science},
  volume  = {447},
  pages   = {85--99},
  year    = {2012},
  publisher={Elsevier}
}

@article{pelc2018reaching,
  title   = {Reaching a Target in the Plane with No Information},
  author  = {Pelc, A.},
  journal = {Information Processing Letters},
  volume  = {140},
  pages   = {13--17},
  year    = {2018},
  publisher={Elsevier}
}

@inproceedings{fricke2016distributed,
  title     = {A Distributed Deterministic Spiral Search Algorithm for Swarms},
  author    = {Fricke, G. M. and Hecker, J. P. and Griego, A. D. and Tran, L. T. and Moses, M. E.},
  booktitle = {2016 IEEE/RSJ International Conference on Intelligent Robots and Systems (IROS)},
  pages     = {4430--4436},
  year      = {2016},
  publisher={IEEE}
}

@inproceedings{bouchard2018deterministic,
  title     = {Deterministic Treasure Hunt in the Plane with Angular Hints},
  author    = {Bouchard, S. and Dieudonn{\'e}, Y. and Pelc, A. and Petit, F.},
  booktitle = {29th International Symposium on Algorithms and Computation (ISAAC)},
  volume    = {123},
  pages     = {48--1},
  year      = {2018},
  publisher={Schloss Dagstuhl--Leibniz-Zentrum fuer Informatik}
}

@article{emek2015many,
  title   = {How Many Ants Does It Take to Find the Food?},
  author  = {Emek, Y. and Langner, T. and Stolz, D. and Uitto, J. and Wattenhofer, R.},
  journal = {Theoretical Computer Science},
  volume  = {608},
  pages   = {255--267},
  year    = {2015},
  publisher={Elsevier}
}

@InProceedings{LangnerKUW15,
  title     = {Overcoming Obstacles with Ants},
  author    = {Langner, T. and Keller, B. and Uitto, J. and Wattenhofer, R.},
  publisher = {Schloss Dagstuhl - Leibniz-Zentrum fuer Informatik},
  year      = {2015},
  booktitle = {International Conference on Principles of Distributed Systems (OPODIS)},
  pages     = {9:1--9:17},
  series    = {LIPIcs},
}

@article{pelc2018information,
  title   = {Information Complexity of Treasure Hunt in Geometric Terrains},
  author  = {Pelc, A. and Yadav, R. N.},
  journal = {arXiv preprint arXiv:1811.06823},
  year    = {2018}
}

@article{pelc2019cost,
  title   = {Cost vs. Information Tradeoffs for Treasure Hunt in the Plane},
  author  = {Pelc, A. and Yadav, R. N.},
  journal = {arXiv preprint arXiv:1902.06090},
  year    = {2019}
}

@book{alpern2006theory,
  title     = {The Theory of Search Games and Rendezvous},
  author    = {Alpern, S. and Gal, S.},
  volume    = {55},
  year      = {2006},
  publisher = {Springer Science \& Business Media}
}

@book{alpern2013search,
  title     = {Search Theory},
  author    = {Alpern, S. and Fokkink, R. and Gasieniec, L. and Lindelauf, R. and Subrahmanian, V. S.},
  year      = {2013},
  publisher = {Springer}
}

@incollection{CGK19search,
  author      = {Czyzowicz, J. and Georgiou, K. and Kranakis, E.},
  title       = {Group Search and Evacuation},
  editor      = {Flocchini, P. and Prencipe, G. and Santoro, N.},
  booktitle   = {Distributed Computing by Mobile Entities; Current Research in Moving and Computing},
  publisher   = {Springer},
  year        = {2019},
  pages       = {335-370},
  chapter     = {14}
}

@book{borodin2005online,
  title     = {Online Computation and Competitive Analysis},
  author    = {Borodin, A. and El-Yaniv, R.},
  year      = {2005},
  publisher = {Cambridge University Press}
}

@inproceedings{Emekicalp2014,
  year      = {2014},
  booktitle = {Proceedings of International Colloquium on Automata, Languages, and Programming (ICALP)},
  title     = {Solving the ANTS Problem with Asynchronous Finite State Machines},
  author    = {Emek, Y. and Langner, T. and Uitto, J. and Wattenhofer, R.},
  pages     = {471-482}
}

@article{bellman1956minimization,
  title   = {Minimization Problem},
  author  = {Bellman, R.},
  journal = {Bull. Amer. Math. Soc},
  volume  = {62},
  number  = {3},
  pages   = {270},
  year    = {1956}
}

@article{chuangpishit2020multi,
  title   = {A Multi-Objective Optimization Problem on Evacuating 2 Robots from the Disk in the Face-to-Face Model; Trade-offs Between Worst-Case and Average-Case Analysis},
  author  = {Chuangpishit, H. and Georgiou, K. and Sharma, P.},
  journal = {Information},
  volume  = {11},
  number  = {11},
  pages   = {506},
  year    = {2020},
  publisher={MDPI}
}

@InProceedings{CzyzowiczGGKMP14,
  title     = {Evacuating Robots via Unknown Exit in a Disk},
  author    = {Czyzowicz, J. and Gasieniec, L. and Gorry, T. and Kranakis, E. and Martin, R. and Pajak, D.},
  booktitle = {DISC 2014},
  publisher = {Springer},
  year      = {2014},
  pages     = {122--136}
}

@InProceedings{CzyzowiczKKNOS17,
  title     = {Linear Search with Terrain-Dependent Speeds},
  author    = {Czyzowicz, J. and Kranakis, E. and Krizanc, D. and Narayanan, L. and Opatrny, J. and Shende, S. M.},
  booktitle = {Algorithms and Complexity - 10th International Conference, CIAC 2017},
  pages     = {430--441},
  year      = {2017}
}

@InProceedings{FeketeGK10,
  title     = {Evacuation of Rectilinear Polygons},
  author    = {Fekete, S. P. and Gray, C. and Kr{\"o}ller, A.},
  booktitle = {Combinatorial Optimization and Applications - 4th International Conference, COCOA 2010},
  pages     = {21--30},
  year      = {2010}
}

@InProceedings{BagheriNO19,
  title     = {Evacuation of Equilateral Triangles by Mobile Agents of Limited Communication Range},
  author    = {Bagheri, I. and Narayanan, L. and Opatrny, J.},
  booktitle = {ALGOSENSORS 2019},
  publisher = {Springer},
  year      = {2019},
  pages     = {3--22}
}

@InProceedings{CzyzowiczKKNOS15,
  title     = {Wireless Autonomous Robot Evacuation from Equilateral Triangles and Squares},
  author    = {Czyzowicz, J. and Kranakis, E. and Krizanc, D. and Narayanan, L. and Opatrny, J. and Shende, S. M.},
  booktitle = {14th International Conference, ADHOC-NOW},
  pages     = {181--194},
  year      = {2015}
}

@inproceedings{georgiou2022triangle,
  title     = {Triangle Evacuation of 2 Agents in the Wireless Model},
  author    = {Georgiou, K. and Jang, W.},
  booktitle = {Algorithmics of Wireless Networks: ALGOSENSORS 2022},
  pages     = {77--90},
  year      = {2022}
}

@article{GLLKllp2023,
  title   = {Evacuating from $\ell_p$ Unit Disks in the Wireless Model},
  author  = {Georgiou, K. and Leizerovich, S. and Lucier, J. and Kundu, S.},
  journal = {Theoretical Computer Science},
  volume  = {944},
  pages   = {113675},
  year    = {2023}
}

@Article{BampasCGIKKP19,
  title   = {Linear Search by a Pair of Distinct-Speed Robots},
  author  = {Bampas, E. and Czyzowicz, J. and Gasieniec, L. and Ilcinkas, D. and Klasing, R. and Kociumaka, T. and Pajak, D.},
  journal = {Algorithmica},
  year    = {2019},
  volume  = {81},
  pages   = {317--342}
}

@article{BGMP2022pfaulty,
  title   = {Algorithms for p-Faulty Search on a Half-Line},
  author  = {Bonato, A. and Georgiou, K. and MacRury, C. and Pra{\l}at, P.},
  journal = {Algorithmica},
  pages   = {1--30},
  year    = {2022}
}

@InProceedings{CzyzowiczGGKKRW17,
  title     = {Evacuation from a Disc in the Presence of a Faulty Robot},
  author    = {Czyzowicz, J. and Georgiou, K. and Godon, M. and Kranakis, E. and Krizanc, D. and Rytter, W. and Wlodarczyk, M.},
  booktitle = {SIROCCO 2017},
  pages     = {158--173},
  year      = {2017}
}

@inproceedings{behrouz2023byzantine,
  title     = {Byzantine Fault-Tolerant Protocols for (n, f)-Evacuation from a Circle},
  author    = {Behrouz, P. and Konstantinidis, O. and Leonardos, N. and Pagourtzis, A. and Papaioannou, I. and Spyrakou, M.},
  booktitle = {International Symposium on Algorithmics of Wireless Networks},
  pages     = {87--100},
  year      = {2023}
}

@book{koopman1946search,
  title     = {Search and Screening},
  author    = {Koopman, B. O.},
  year      = {1946},
  publisher = {Operations Evaluation Group}
}

@book{stone1975theory,
  title     = {Theory of Optimal Search},
  author    = {Stone, L. D.},
  year      = {1975},
  publisher = {Academic Press}
}

@inproceedings{feinerman2012collaborative,
  title     = {Collaborative Search on the Plane Without Communication},
  author    = {Feinerman, O. and Korman, A. and Lotker, Z. and Sereni, J.-S.},
  booktitle = {Proceedings of the 2012 ACM Symposium on Principles of Distributed Computing},
  pages     = {77--86},
  year      = {2012}
}

@article{baeza1993searching,
  title   = {Searching in the Plane},
  author  = {Baeza-Yates, R. A. and Culberson, J. C. and Rawlins, G. J. E.},
  journal = {Information and Computation},
  volume  = {106},
  pages   = {234--252},
  year    = {1993}
}

@article{papadimitriou1991shortest,
  title   = {Shortest Paths Without a Map},
  author  = {Papadimitriou, C. H. and Yannakakis, M.},
  journal = {Theoretical Computer Science},
  volume  = {84},
  pages   = {127--150},
  year    = {1991}
}

@article{angelopoulos2011multi,
  title   = {Multi-Target Ray Searching Problems},
  author  = {Angelopoulos, S. and L{\'o}pez-Ortiz, A. and Panagiotou, K.},
  journal = {Theoretical Computer Science},
  volume  = {540},
  pages   = {2--12},
  year    = {2014}
}

@article{demaine2006online,
  title   = {Online Searching with Turn Cost},
  author  = {Demaine, E. D. and Fekete, S. P. and Gal, S.},
  journal = {Theoretical Computer Science},
  volume  = {361},
  pages   = {342--355},
  year    = {2006}
}

@inproceedings{beauquier2010utilizing,
  title     = {On Utilizing Speed in Networks of Mobile Agents},
  author    = {Beauquier, J. and Burman, J. and Clement, J. and Kutten, S.},
  booktitle = {Proceedings of the 29th ACM SIGACT-SIGOPS Symposium on Principles of Distributed Computing},
  pages     = {305--314},
  year      = {2010}
}

@article{Albers00,
  title   = {Exploring Unknown Environments},
  author  = {Albers, S. and Henzinger, M. R.},
  journal = {SIAM Journal on Computing},
  volume  = {29},
  pages   = {1164-1188},
  year    = {2000}
}

@article{DDKPU13,
  title   = {Fast Collaborative Graph Exploration},
  author  = {Dereniowski, D. and Disser, Y. and Kosowski, A. and Pajak, D. and Uzna{\'n}ski, P.},
  journal = {Information and Computation},
  volume  = {243},
  pages   = {37--49},
  year    = {2015}
}

@incollection{AS11,
  title     = {Online Algorithms - What Is It Worth to Know the Future?},
  author    = {Albers, S. and Schmelzer, S.},
  booktitle = {Algorithms Unplugged},
  pages     = {361-366},
  year      = {2011}
}

@article{Albers03,
  title   = {Online Algorithms: A Survey},
  author  = {Albers, S.},
  journal = {Mathematical Programming},
  volume  = {97},
  pages   = {3-26},
  year    = {2003}
}

@article{HKLT13,
  title   = {Online Graph Exploration Algorithms for Cycles and Trees by Multiple Searchers},
  author  = {Higashikawa, Y. and Katoh, N. and Langerman, S. and Tanigawa, S.},
  journal = {Journal of Combinatorial Optimization},
  volume  = {28},
  pages   = {480-495},
  year    = {2014}
}

@article{WIFBZP11,
  title   = {Optimizing Sensor Movement Planning for Energy Efficiency},
  author  = {Wang, G. and Irwin, M. J. and Fu, H. and Berman, P. and Zhang, W. and La Porta, T.},
  journal = {ACM Transactions on Sensor Networks},
  volume  = {7},
  pages   = {33},
  year    = {2011}
}

@incollection{CGKK11,
  title     = {Boundary Patrolling by Mobile Agents with Distinct Maximal Speeds},
  author    = {Czyzowicz, J. and Gasieniec, L. and Kosowski, A. and Kranakis, E.},
  booktitle = {Algorithms -- ESA 2011},
  pages     = {701-712},
  year      = {2011}
}

@article{KK12,
  title   = {Fence Patrolling by Mobile Agents with Distinct Speeds},
  author  = {Kawamura, A. and Kobayashi, Y.},
  journal = {Distributed Computing},
  volume  = {28},
  pages   = {147--154},
  year    = {2015}
}

@incollection{CFMS10,
  title     = {Network Exploration by Silent and Oblivious Robots},
  author    = {Chalopin, J. and Flocchini, P. and Mans, B. and Santoro, N.},
  booktitle = {Graph Theoretic Concepts in Computer Science},
  pages     = {208--219},
  year      = {2010}
}

@article{DFKNS07,
  title   = {Map Construction of Unknown Graphs by Multiple Agents},
  author  = {Das, S. and Flocchini, P. and Kutten, S. and Nayak, A. and Santoro, N.},
  journal = {Theoretical Computer Science},
  volume  = {385},
  pages   = {34-48},
  year    = {2007}
}

@article{Bellman63,
  title   = {An Optimal Search},
  author  = {Bellman, R.},
  journal = {SIAM Review},
  volume  = {5},
  pages   = {274-274},
  year    = {1963}
}

@article{Beck64,
  title   = {On the Linear Search Problem},
  author  = {Beck, A.},
  journal = {Israel Journal of Mathematics},
  volume  = {2},
  pages   = {221--228},
  year    = {1964}
}

@article{CILP13,
  title   = {Worst-Case Optimal Exploration of Terrains with Obstacles},
  author  = {Czyzowicz, J. and Ilcinkas, D. and Labourel, A. and Pelc, A.},
  journal = {Information and Computation},
  volume  = {225},
  pages   = {16--28},
  year    = {2013}
}

@article{DKP91,
  title   = {How to Learn an Unknown Environment. I: The Rectilinear Case},
  author  = {Deng, X. and Kameda, T. and Papadimitriou, C.},
  journal = {Journal of the ACM},
  volume  = {45},
  pages   = {215--245},
  year    = {1998}
}

@misc{lopez2015optimal,
  title   = {Optimal Strategies for Search and Rescue Operations with Robot Swarms},
  author  = {L{\'o}pez-Ortiz, A. and Maftuleac, D.},
  journal = {arXiv preprint arXiv:1410.1077},
  year    = {2015}
}

@InProceedings{ConleyGeorgiou25Inspection,
author="Conley, J.
and Georgiou, K.",
editor="Schmid, U.
and Kuznets, R.",
title="Multi-agent Disk Inspection",
booktitle="Structural Information and Communication Complexity",
year="2025",
publisher="Springer Nature Switzerland",
address="Cham",
pages="262--280",
}

@misc{temerev2026lower,
  author = {Alexander Temerev},
  title = {A New Unconditional Lower Bound for Shoreline Search},
  year = {2026},
  eprint = {2608.18362},
  archivePrefix = {arXiv}
}

@misc{temerev2026optimal,
  author = {Alexander Temerev},
  title = {The Logarithmic Spiral Is Optimal for Shoreline Search: A Computer-Assisted Proof},
  year = {2026},
  eprint = {2609.24454},
  archivePrefix = {arXiv}
}

\appendix

\section{Technical Calculations Within the Proof of Theorem~\ref{thm: numerical upper bound}}
\label{app: numerical upper bound}

\begin{lemma}
\label{lem: properties of C phi}
For $0<\phi<\pi/2$, let
$
C(\phi)=\cos(\phi/2)\mathcal F_\phi.
$
Then $C(\phi)$ is differentiable and strictly increasing, and
$$
\frac{C'(\phi)}{C(\phi)}
\geq
\frac{2}{\sqrt5}.
$$
Moreover,
$
\lim_{\phi\rightarrow0^+}C(\phi)=1
$,
and
$
\lim_{\phi\rightarrow(\pi/2)^-}C(\phi)=\infty.
$
\end{lemma}

\begin{proof}
First we consider the cone minimum.
Let
$$
f_\phi(\theta)=\cos(\phi/2)
\left(\frac{\sin\theta}{\sin(\theta-\phi)}\right)^2
\frac1{\cos(\theta-\phi/2)},
$$
for $\phi<\theta\leq\pi/2$.
Since $f_\phi(\theta)>0$, its critical points are also the critical points of $\log f_\phi(\theta)$. We therefore calculate
$$
\frac{\partial}{\partial\theta}\log f_\phi(\theta)
=
2\cot\theta-2\cot(\theta-\phi)+\tan(\theta-\phi/2),
$$
and after basic trigonometric simplifications, 
$$
\partial_\theta^2\log f_\phi
=2\bigl(\csc^2(\theta-\phi)-\csc^2\theta\bigr)
 +\sec^2(\theta-\phi/2)>0.
$$
The first derivative tends to $-\infty$ as $\theta\rightarrow\phi^+$, and at $\theta=\pi/2$ it equals $\cot(\phi/2)-2\tan\phi$. Thus, with
$$
\phi_{\mathrm c}:=\arccos(2/3),
$$
the minimum is attained at a unique interior point for $0<\phi<\phi_{\mathrm c}$, and at $\theta=\pi/2$ for $\phi_{\mathrm c}\leq\phi<\pi/2$.

The change of variable $\beta=\sin\theta/\sin(\theta-\phi)$ maps $(\phi,\pi/2]$ bijectively onto $[\sec\phi,\infty)$, in reverse order. By the cost formula of Lemma~\ref{lem: choice of sequence g},
$$
C(\phi)=\min_{\beta\geq\sec\phi}
\frac{\beta^2\sqrt{\beta^2-2\beta\cos\phi+1}}{\beta-1}.
$$
For the interior minimum, set $u:=\cos\phi$. The optimality condition with respect to $\beta$ gives
$
2\beta^3-3(1+u)\beta^2+(1+5u)\beta-2=0,
$
and solving it for $u$, we obtain
$$
u=1+\frac{2(\beta-1)^3}{\beta(3\beta-5)}.
$$
Since $u<1$, this implies $\beta<5/3$. Furthermore,
$$
\beta u-1=
\frac{(\beta-1)(\beta+1)(2\beta-3)}{3\beta-5}\geq0,
$$
so $1<\beta\leq3/2$. 
Differentiating the expression for $u$ with respect to $\beta$, we obtain
$$
\frac{\dd u}{\dd\beta}
=
\frac{2(\beta-1)^2(3\beta^2-4\beta-5)}
{\beta^2(3\beta-5)^2}
<0,
$$
for $1<\beta<3/2$. Therefore, the inverse function $\beta=\beta(u)$ is differentiable, and hence so is $\beta=\beta(\phi)$ for $0<\phi<\phi_{\mathrm c}$.

We next derive a lower bound on $C'(\phi)/C(\phi)$, which will be used below to establish the monotonicity of $C(\phi)R(\phi)$ and $\gamma(\phi)=C(\phi)/R(\phi)$. Let $\beta=\beta(\phi)$ denote the minimizing value at the interior minimum. Then
$$
\log C(\phi)
=
2\log\beta
+
\frac{1}{2}\log\left(\beta^2-2\beta\cos\phi+1\right)
-
\log(\beta-1).
$$
Differentiating with respect to $\phi$, we obtain
$$
\frac{C'(\phi)}{C(\phi)}
=
\beta'(\phi)
\left(
\frac{2}{\beta}
+
\frac{\beta-\cos\phi}{\beta^2-2\beta\cos\phi+1}
-
\frac{1}{\beta-1}
\right)
+
\frac{\beta\sin\phi}{\beta^2-2\beta\cos\phi+1}.
$$
The expression multiplying $\beta'(\phi)$ is equal to $0$ at the minimizing value of $\beta$. Therefore,
$$
\frac{C'(\phi)}{C(\phi)}
=
\frac{\beta\sin\phi}{\beta^2-2\beta\cos\phi+1}
>0.
$$
To express this quantity only in terms of $\beta$, we use $u=\cos\phi$. Squaring the above expression, using $\sin^2\phi=1-u^2$, and substituting
$
u=1+\frac{2(\beta-1)^3}{\beta(3\beta-5)},
$
we obtain
$$
\left(\frac{C'(\phi)}{C(\phi)}\right)^2
=
\frac{4(-\beta^2+\beta+1)}{\beta^2-1}.
$$
The expression on the right is decreasing for $1<\beta\leq3/2$, and at $\beta=3/2$ it equals $4/5$. Since $C'(\phi)/C(\phi)>0$, it follows that
$$
\frac{C'(\phi)}{C(\phi)}
\geq
\frac{2}{\sqrt5}.
$$
We now consider the case in which the minimum is attained at $\theta=\pi/2$. Substituting $\theta=\pi/2$ gives
$
C(\phi)=\cot(\phi/2)\sec^2\phi,
$
and hence
$$
\frac{C'(\phi)}{C(\phi)}
=
2\tan\phi-\frac1{\sin\phi}.
$$
This expression is increasing in $\phi$, and at $\phi=\phi_{\mathrm c}$ it equals $2/\sqrt5$. As $\phi$ approaches $\phi_{\mathrm c}$ from below, the interior minimizing value of $\theta$ approaches $\pi/2$, and therefore $\beta$ approaches
$
\sec\phi_{\mathrm c}=\frac32.
$
Hence, the expression for $C'(\phi)/C(\phi)$ obtained for the interior minimum also approaches $2/\sqrt5$. Thus, $C(\phi)$ is differentiable at $\phi=\phi_{\mathrm c}$, and
$$
\frac{C'(\phi)}{C(\phi)}
\geq
\frac{2}{\sqrt5},
$$
for all $0<\phi<\pi/2$. 
Finally,
$
\lim_{\phi\rightarrow 0^+}C(\phi)=1.
$
Indeed, the expression in terms of $\beta$ gives $C(\phi)\geq1$, while choosing $\theta=\sqrt\phi$ gives the corresponding upper bound tending to $1$ as $\phi\rightarrow0^+$. Also, from
$$
C(\phi)=\cot(\phi/2)\sec^2\phi
$$
for $\phi\geq\phi_{\mathrm c}$, we obtain
$
\lim_{\phi\rightarrow(\pi/2)^-}C(\phi)=\infty.
$
\end{proof}

\begin{lemma}
\label{lem: properties of R phi}
For $0<\phi<\pi/2$, let
$$
R(\phi)=\mathcal R_\phi.
$$
Then, for every $0\leq\phi_1<\phi_2<\pi/2$,
$$
e^{-2(\phi_2-\phi_1)/(3\pi)}R(\phi_1)
<
R(\phi_2)
<
R(\phi_1),
$$
where we set $R(0):=\mathcal U_1$. In particular, $R(\phi)$ is continuous and strictly decreasing,
$$
\lim_{\phi\rightarrow0^+}R(\phi)=\mathcal U_1,
$$
and $R(\phi)>1$ for $0<\phi<\pi/2$.
\end{lemma}

\begin{proof}
Set $s=\sin(\phi/2)$ and make the change of variable
$$
a=\frac{\lambda-1}{2\sqrt\lambda\,s},
$$
which is a bijection from $\lambda\in(1,\infty)$ to $a\in(0,\infty)$, with inverse
$
\lambda=\exp\left(2\operatorname{arsinh}(as)\right).
$
(and as a reminder, $\operatorname{arsinh}x=\ln(x+\sqrt{1+x^2})$). 
The objective of Theorem~\ref{thm: wedge search upper bound} becomes
$$
Q(\phi,a)=\frac{\sqrt{1+a^2}}a
\exp\left(2\operatorname{arsinh}(a\sin(\phi/2))+(2\pi-\phi)a\right),
$$
and $R(\phi)=\min_{a>0}Q(\phi,a)$.
Differentiating $Q(\phi,a)$ with respect to $a$ and setting the derivative equal to zero gives
$$
\frac1{1+a^2}
=
(2\pi-\phi)a+\frac{2as}{\sqrt{1+a^2s^2}}
$$
The left-hand side strictly decreases from $1$ to $0$, and the right-hand side strictly increases from $0$ to infinity. Since $Q(\phi,a)$ tends to infinity as $a\rightarrow0^+$ and as $a\rightarrow\infty$, there is a unique global minimizer $a=a_\phi$.
Substituting $a=a_\phi$ in the above equation, and observing that the second term on the right-hand side is positive while the left-hand side is smaller than $1$, we obtain
$$
0<a_\phi<\frac1{2\pi-\phi}<\frac2{3\pi}.
$$
For $\phi=0$, the expression for $Q(\phi,a)$ becomes
$$
Q(0,a)=\frac{\sqrt{1+a^2}}a e^{2\pi a}.
$$
Therefore,
$
\min_{a>0}Q(0,a)=\mathcal U_1,
$
and the minimizing value is $a=\kappa_1$. We therefore set $a_0:=\kappa_1$. In particular, $\kappa_1<1/(2\pi)$.
For fixed $a>0$, differentiating $\log Q(\phi,a)$ with respect to $\phi$ gives
$$
\frac{\partial}{\partial\phi}\log Q(\phi,a)
=
a\left(
\frac{\cos(\phi/2)}{\sqrt{1+a^2\sin^2(\phi/2)}}-1
\right).
$$
Since
$$
0<
\frac{\cos(\phi/2)}{\sqrt{1+a^2\sin^2(\phi/2)}}
<1
$$
for $0<\phi<\pi/2$, it follows that
$$
-a<
\frac{\partial}{\partial\phi}\log Q(\phi,a)
<0.
$$
Let $0\leq\phi_1<\phi_2<\pi/2$. For $0<a<2/(3\pi)$, the above derivative bounds give
$$
e^{-2(\phi_2-\phi_1)/(3\pi)}Q(\phi_1,a)
<Q(\phi_2,a)<Q(\phi_1,a).
$$
Applying the first inequality with $a=a_{\phi_2}$ and the second with $a=a_{\phi_1}$, and using the definition of $R$, we obtain
$$
e^{-2(\phi_2-\phi_1)/(3\pi)}R(\phi_1)
<R(\phi_2)<R(\phi_1).
$$
Thus, $R$ is continuous and strictly decreasing, and
$
\lim_{\phi\rightarrow0^+}R(\phi)=\mathcal U_1.
$
Also, $R(\phi)>1$, since its minimum is attained and $Q(\phi,a)>1$ for $a>0$.
\end{proof}

\end{document}